	\documentclass[
		paper=a4,11pt,american,pagesize,%
		abstract=true,DIV=12,headinclude=true,headlines=0,footinclude=true,footlines=-5,twoside=semi,%
	]{scrartcl}
	\def\docclass{koma}
	\def\version{arxiv}

	\def\draftmode{false} 

\usepackage[utf8]{inputenc}
\makeatletter
\usepackage{xifthen}

\newcommand\iflipics[2]{\ifthenelse{\equal{\docclass}{lipics}}{#1}{#2}}
\newcommand\ifkoma[2]{\ifthenelse{\equal{\docclass}{koma}}{#1}{#2}}
\newcommand\ifieee[2]{\ifthenelse{\equal{\docclass}{ieee}}{#1}{#2}}
\newcommand\ifsiam[2]{\ifthenelse{\equal{\docclass}{siam}}{#1}{#2}}
\newcommand\ifsiamsingle[2]{\ifthenelse{\equal{\docclass}{siam-single}}{#1}{#2}}
\newcommand\ifmysiam[2]{\ifthenelse{\equal{\docclass}{my-siam}}{#1}{#2}}
\newcommand\ifacm[2]{\ifthenelse{\equal{\docclass}{acm}}{#1}{#2}}
\newcommand\ifdcc[2]{\ifthenelse{\equal{\docclass}{dcc}}{#1}{#2}}
\newcommand\ifspringerjournal[2]{\ifthenelse{\equal{\docclass}{springer-journal}}{#1}{#2}}
\ifthenelse{ \equal{\docclass}{lipics} \OR \equal{\docclass}{koma} \OR \equal{\docclass}{ieee} \OR \equal{\docclass}{siam} \OR \equal{\docclass}{siam-single} \OR \equal{\docclass}{my-siam} \OR \equal{\docclass}{acm} \OR \equal{\docclass}{dcc} \OR \equal{\docclass}{springer-journal} }{
	\PackageInfo{paper}{Building paper with docclass = \docclass} 
}{
	\PackageWarning{paper}{docclass = "\docclass", but must be one of "lipics", "koma", "ieee", "siam", "siam-single", "my-siam", "acm", "dcc", "springer-journal"}
}

\newcommand\ifmanuscript[2]{\ifthenelse{\equal{\version}{manuscript}}{#1}{#2}}
\newcommand\ifarxiv[2]{\ifthenelse{\equal{\version}{arxiv}}{#1}{#2}}
\newcommand\ifsubmission[2]{\ifthenelse{\equal{\version}{submission}}{#1}{#2}}
\newcommand\ifproceedings[2]{\ifthenelse{\equal{\version}{proceedings}}{#1}{#2}}
\ifthenelse{ 
	\equal{\version}{manuscript} 
	\OR \equal{\version}{arxiv} 
	\OR \equal{\version}{submission} 
	\OR \equal{\version}{proceedings} 
}{
	\PackageInfo{paper}{Building paper version = \version} 
}{
	\PackageWarning{paper}{version = "\version", but must be one of "manuscript", "arxiv", "submission", "proceedings"}
}

\newcommand\ifdraft[2]{\ifthenelse{\equal{\draftmode}{true}}{#1}{#2}}
\ifthenelse{ \equal{\draftmode}{true} \OR \equal{\draftmode}{false} }{
	\PackageInfo{paper}{Building paper with draftmode = \draftmode} 
}{
	\PackageWarning{paper}{draftmode = "\draftmode", but must be "true" or "false"}
}

\usepackage[T1]{fontenc}
\ifsiam{
	\usepackage{lmodern}
	\usepackage{slantsc}
}{}
\ifsiamsingle{
	\usepackage{lmodern}
	\usepackage{slantsc}
}{}
\ifmysiam{
	\usepackage{lmodern}
	\usepackage{slantsc}
}{}
\ifkoma{
	\usepackage{lmodern}
	\usepackage{slantsc}
}{}
\iflipics{
	\usepackage{lmodern}
	\usepackage{slantsc}
}{}
\ifdcc{
	\usepackage{lmodern}
	\usepackage{slantsc}
}{}
\ifspringerjournal{
	\usepackage{lmodern}
	\usepackage{slantsc}
	\usepackage{xcolor}
}{}

\usepackage{babel}

\usepackage{array,multicol}
\ifieee{
	\usepackage[cmex10]{amsmath,mathtools}
	\usepackage{amsfonts,amssymb}
}{}
\ifkoma{
	\usepackage{amsmath,amsfonts,amssymb,mathtools}
}{}
\iflipics{
	\usepackage{amsmath,amsfonts,amssymb,mathtools}
}{}
\ifsiam{
	\usepackage{amsmath,amsfonts,amssymb,mathtools}
}{}
\ifsiamsingle{
	\usepackage{amsmath,amsfonts,amssymb,mathtools}
}{}
\ifmysiam{
	\usepackage{amsmath,amsfonts,amssymb,mathtools}
}{}
\ifacm{
	\usepackage{mathtools}
}{}
\ifdcc{
	\usepackage{amsmath,amsfonts,amssymb,mathtools}
}{}
\ifspringerjournal{
	\usepackage{amsmath,amsfonts,amssymb,mathtools}
}{}

\usepackage{mleftright}\mleftright 
\usepackage{relsize,xspace,booktabs,adjustbox,needspace,pbox,relsize}
\ifieee{
		
		\usepackage{enumitem}
}{
	\ifspringerjournal{
		\usepackage{enumitem}
		\setlist{topsep=\medskipamount}
	}{
		\usepackage{enumitem}
	}
}
\usepackage{graphicx}

\ifacm{}{
	\usepackage{colonequals}
}

\usepackage[american]{wref}

\ifsiam{
	\usepackage{ltexpprt}
}{}
\ifsiamsingle{
	\usepackage{ltexpprt}
}{}

\newdimen\makeboxdimen

\newcommand\plaincenter[1]{%
	\mbox{}\hfill#1\hfill\mbox{}%
}

\ifthenelse{\equal{\docclass}{koma} \OR \equal{\docclass}{my-siam}}{
	\setlength\parindent{1.5em}
	\usepackage[headsepline]{scrlayer-scrpage}
	\pagestyle{scrheadings}
	\clearscrheadfoot
	\AtBeginDocument{%
		\automark[section]{}%
	}
	\ohead{\pagemark}
	\rehead{\mytitle}
	\lohead{\headmark}
	\addtokomafont{caption}{\sffamily\small}
	\addtokomafont{captionlabel}{\sffamily\textbf}
	\setcapmargin{2em}
}{}
\ifmysiam{
	\setcapmargin{1em}
	\setcapindent{0em}
}{}
\ifmysiam{
	\setlength\parskip{0pt}
	\RedeclareSectionCommand[
		beforeskip=-1.25\baselineskip,
		afterskip=0.75\baselineskip,
	]{section}
	\RedeclareSectionCommand[
		beforeskip=-1\baselineskip,
		afterskip=-1.5em,
	]{subsection}
	\RedeclareSectionCommand[
		beforeskip=-1\baselineskip,
		afterskip=-1.5em,
	]{subsubsection}
	\RedeclareSectionCommand[
		beforeskip=-.25\baselineskip,
		indent=1.5em,
		afterskip=-1em,
	]{paragraph}
}{}
\ifdcc{
	\ifproceedings{}{
		\pagestyle{plain}
		\setlength{\footskip}{5ex}
	}
}{}

\AtBeginDocument{%
	\let\mytitle\@title%
}

\ifsiam{
	\usepackage[bibtex-alpha]{url-doi-arxiv}
}{}
\ifsiamsingle{
	\usepackage[bibtex-alpha]{url-doi-arxiv}
}{}
\ifmysiam{
	\usepackage[bibtex-alpha]{url-doi-arxiv}
}{}
\ifkoma{
	\usepackage[bibtex-alpha]{url-doi-arxiv}
}{}
\iflipics{
	\usepackage[bibtex]{url-doi-arxiv}
}{}
\ifdcc{
	\usepackage[bibtex-alpha]{url-doi-arxiv}
}{}
\ifspringerjournal{
	\usepackage[bibtex-alpha]{url-doi-arxiv}
}{}

\let\oldthebibliography\thebibliography
\renewcommand\thebibliography[1]{%
	\oldthebibliography{#1}%
	\pdfbookmark[1]{References}{}%
}

\usepackage{lscape} 

\ifkoma{
	\usepackage{float}
	\floatstyle{plain}

}{}
\iflipics{
	\usepackage{newfloat}
	\DeclareFloatingEnvironment[%
			name=Algorithm,%
			placement=thb,%
		]{algorithm}
}{}
\ifmysiam{
	\usepackage{newfloat}
	\DeclareFloatingEnvironment[%
			name=Algorithm,%
			placement=thb,%
		]{algorithm}
}

\ifdcc{
	\usepackage{float}
	\usepackage[font={small},labelfont=bf]{caption}
}{}
\ifmysiam{
	\setcounter{topnumber}{3}
	\setcounter{bottomnumber}{3}
	\setcounter{totalnumber}{3}     
	\setcounter{dbltopnumber}{3}    

}{}
\ifsiam{

	\setcounter{topnumber}{3}
	\setcounter{bottomnumber}{3}
	\setcounter{totalnumber}{3}     
	\setcounter{dbltopnumber}{3}    

}{}
\ifsiamsingle{

	\setcounter{topnumber}{3}
	\setcounter{bottomnumber}{3}
	\setcounter{totalnumber}{3}     

}{}

\usepackage{dcolumn}

\usepackage{textcomp} 
\usepackage{listings}

\usepackage{tikz}

\usetikzlibrary{positioning,arrows.meta,fit}
\usetikzlibrary{backgrounds,calc,trees,graphs}
\usetikzlibrary{shapes.geometric,shapes.misc}
\usetikzlibrary{decorations.pathreplacing}

\pgfdeclarelayer{background}
\pgfsetlayers{background,main}

\usetikzlibrary{external}
\tikzsetexternalprefix{pics/externalized/}
\tikzset{
	external/system call={%
		lualatex \tikzexternalcheckshellescape -halt-on-error %
			-interaction=batchmode -jobname "\image" "\texsource"%
	},
}
\tikzset{external/export=false} 

\iflipics{
	\newtheorem{fact}[theorem]{Fact}

	\newenvironment{proofof}[1]{%
		\begin{proof}[{{Proof of #1{}}}]%
	}{%
		\end{proof}%
	}
}{}
\ifacm{
	\AtEndPreamble{
		\theoremstyle{acmdefinition}
		\newtheorem{remark}[theorem]{Remark}
		\newtheorem{fact}[theorem]{Fact}
	}
	
	\newenvironment{proofof}[1]{%
		\begin{proof}[{{Proof of #1{}}}]%
	}{%
		\end{proof}%
	}
}{}
\ifsiam{
	\newtheorem{remark}{Remark}
	\newenvironment{proofof}[1]{%
			\begin{proof}[{{#1{}}}]%
		}{%
			\end{proof}%
		}
}{}
\ifsiamsingle{
	\newenvironment{proofof}[1]{%
			\begin{proof}[{{#1{}}}]%
		}{%
			\end{proof}%
		}

	\newtheorem{theorem}{Theorem}[section]
	\newtheorem{proposition}[theorem]{Proposition}
	\newtheorem{lemma}[theorem]{Lemma}
	\newtheorem{conjecture}[theorem]{Conjecture}
	\newtheorem{corollary}[theorem]{Corollary}
	\newtheorem{definition}[theorem]{Definition}
	\newtheorem{remark}[theorem]{Remark}
	\newtheorem{claim}[theorem]{Claim}
	\newtheorem{question}[theorem]{Question}
}{}
\ifthenelse{%
		\equal{\docclass}{lipics} \OR \equal{\docclass}{siam} \OR 
		\equal{\docclass}{siam-single} \OR \equal{\docclass}{acm}%
}{}{
	\usepackage[amsmath,hyperref,thmmarks]{ntheorem}
	
	\theorembodyfont{\slshape}
	\theoremseparator{:}
	\newtheoremstyle{proofstyle}%
	  {\item[\theorem@headerfont\hskip\labelsep ##1\theorem@separator]}%
	  {\item[\theorem@headerfont\hskip\labelsep ##3\theorem@separator]}
	
	\theorempreskip{\topsep} 

	\theoremsymbol{\adjustbox{scale=.8}{$\triangleleft\mkern-1mu$}}
	
	\newtheorem{theorem}{Theorem}[section]
	
	\theoremstyle{plain}
	\theorempreskip{\topsep}
	
	\newtheorem{proposition}[theorem]{Proposition}
	\newtheorem{lemma}[theorem]{Lemma}
	\newtheorem{conjecture}[theorem]{Conjecture}
	\newtheorem{corollary}[theorem]{Corollary}
	\newtheorem{definition}[theorem]{Definition}
	\newtheorem{question}[theorem]{Question}
	
	\theoremstyle{plain}
	\theorembodyfont{\upshape}

	\newtheorem{remark}[theorem]{Remark}
	
	\newtheorem{claim}[theorem]{Claim}
	
	\theoremsymbol{\raisebox{-.25ex}{$\Box$}}
	\qedsymbol{\raisebox{-.25ex}{$\Box$}}
	
	\theoremstyle{proofstyle}
	\newtheorem{proof}{Proof}
	\newenvironment{proofof}[1]{%
		\begin{proof}[{{Proof of #1{}}}]%
	}{%
		\end{proof}%
	}

}

\iflipics{
		\newenvironment{thmenumerate}[2][]{%
			\begin{enumerate}[
				label={\textsf{\textbf{\color{darkgray}{\makebox[\widthof{(a)}][c]{\textup{(\alph*)}}}}}},
				ref={\ref{#2}\kern.1em--\kern.1em(\alph*)},
				itemsep=0pt,
				topsep=.5ex,
				leftmargin=1.75em,
				#1
			]%
		}{%
			\end{enumerate}%
		}
}{
	\ifspringerjournal{
		\newenvironment{thmenumerate}[2][]{%
			\begin{enumerate}[
				label={\makebox[\widthof{(a)}][c]{\textup{(\alph*)}}},
				ref={\ref{#2}\kern.1em--\kern.1em(\alph*)},
				itemsep=0pt,
				topsep=\smallskipamount,
				leftmargin=1.75em,
				#1
			]%
		}{%
			\end{enumerate}%
		}
	}{
		\newenvironment{thmenumerate}[2][]{%
			\begin{enumerate}[
				label={\makebox[\widthof{(a)}][c]{\textup{(\alph*)}}},
				ref={\ref{#2}\kern.1em--\kern.1em(\alph*)},
				itemsep=0pt,
				#1
			]%
		}{%
			\end{enumerate}%
		}
	}
}

\newcommand*\ie{\mbox{i.\hspace{.2ex}e.}}
\newcommand*\eg{\mbox{e.\hspace{.2ex}g.}}

\newcommand\N{\mathbb N}

\usepackage{fixmath}

\newcommand{\ESymbol}{\mathbb{E}}

\newcommand{\ProbSymbol}{\ensuremath{\mathbb{P}}}

\DeclarePairedDelimiterXPP\Prob[1]{\ProbSymbol}[]{}{%
	#1%
}
\DeclarePairedDelimiterXPP\E[1]{\ESymbol}[]{}{%
	#1%
}
\DeclarePairedDelimiterXPP\Eover[2]{\ESymbol_{#1}}[]{}{%
	#2%
}
\DeclarePairedDelimiterXPP\ProbIn[2]{\ProbSymbol_{#1}}[]{}{%
	#2%
}
\providecommand{\Prob}{} 
\providecommand{\ProbIn}{} 
\providecommand{\E}{} 
\providecommand{\Eover}{} 

\newcommand{\surroundedmath}[3]{
	\mathchoice{
		#1{#2{#3}#2}%
	}{
		#1{#3}%
	}{
		#1{#3}%
	}{
		#1{#3}%
	}%
}

\newcommand\wrel[1]{\surroundedmath{\mathrel}{\;}{#1}}
\newcommand\wwrel[1]{\surroundedmath{\mathrel}{\;\;}{#1}}

\ifthenelse{%
		\equal{\docclass}{lipics} \OR \equal{\docclass}{siam-single}%
}{}{
	\makeatletter
	\let\oldalign\align
	\let\endoldalign\endalign
	\renewenvironment{align}{%
		\begingroup%
		\let\oldhalign\halign
		\def\halign{%
			\let\oldbreak\\%
			\def\nonnumberbreak{\nonumber\oldbreak*}%
			\def\\{%
				\@ifstar{\nonnumberbreak}{\oldbreak}%
			}%
			\oldhalign%
		}
		\oldalign%
	}{%
		\endoldalign%
		\endgroup%
	}
}
\newcommand*\numberthis[1][]{\stepcounter{equation}\tag{\theequation}}

\allowdisplaybreaks[3]

\newcommand\splitaftercomma[1]{%
  \begingroup
  \begingroup\lccode`~=`, \lowercase{\endgroup
    \edef~{\mathchar\the\mathcode`, \penalty0 \noexpand\hspace{0pt plus .25em}}%
  }\mathcode`,="8000 #1%
  \endgroup
}

\def\mydots{\xleaders\hbox to.5em{\hfill.\hfill}\hfill}
\newlength\tmpLenNotations

\ifdraft{%
	\iflipics{
		\usepackage{lineno} 
	}{
		\usepackage[switch]{lineno} 
	}
	\linenumbers
	\overfullrule=6mm
	
	\usepackage[color,notref,notcite]{showkeys}
	\definecolor{refkey}{gray}{.99}
	\colorlet{labelkey}{green!60!black!60}
	
	\usepackage[inline,nolabel]{showlabels}

	\showlabels{cite}
	\showlabels{citealt}
	\showlabels{citealp}
	\showlabels{citet}
	\showlabels{Citet}
	\showlabels{citep}
	\showlabels{citeauthor}
	\showlabels{Citeauthor}
	\showlabels{citefullauthor}
	\showlabels{citeyear}
	\showlabels{citeyearpar}
	\showlabels{wref}
	\showlabels{wpref}
	\showlabels{wtpref}
	\showlabels{wildref}
	\showlabels{wildpageref}
	\showlabels{wildtpageref}
}{}

\iflipics{
	\ifmanuscript{\hideLIPIcs}{}
	\ifarxiv{\hideLIPIcs}{}
	\ifsubmission{}{\nolinenumbers}
}{}

\ifdraft{}{%
	\usepackage{microtype}
}

\hypersetup{
	final,
	unicode=true, 
	bookmarks=true,
	bookmarksnumbered=true,
	bookmarksdepth=2,
	bookmarksopen=true,
	breaklinks=true,
	hidelinks,
}

\newsavebox\tmpbox

\ifdcc{
	\renewcommand\paragraph{\@startsection{paragraph}{4}{\parindent}
	                                      {\smallskipamount}
	                                      {-1em}%
	                                      {\normalfont\normalsize\bfseries}}
}{}
\iflipics{
	\let\oldparagraph\paragraph
	\renewcommand\paragraph[1]{%
		\oldparagraph*{#1}
	}
}{
	\let\oldparagraph\paragraph
	\renewcommand\paragraph[1]{%
		\oldparagraph{#1.}
	}
}

\ifmysiam{
	\let\oldsubsection\subsection
	\renewcommand\subsection[1]{%
		\oldsubsection{#1.}%
	}
	\let\oldsubsubsection\subsubsection
	\renewcommand\subsubsection[1]{%
		\oldsubsubsection{#1.}%
	}
}{}
\ifsiam{
	\let\oldsubsection\subsection
	\renewcommand\subsection[1]{%
		\oldsubsection{#1.}%
	}
	\let\oldsubsubsection\subsubsection
	\renewcommand\subsubsection[1]{%
		\oldsubsubsection{#1.}%
	}
}{}
\ifsiamsingle{
	\let\oldsubsection\subsection
	\renewcommand\subsection[1]{%
		\oldsubsection{#1.}%
	}
	\let\oldsubsubsection\subsubsection
	\renewcommand\subsubsection[1]{%
		\oldsubsubsection{#1.}%
	}
}{}

\let\epsilon\varepsilon

\def\myacknowledgements{}
\ifkoma{
	
}{}
\ifieee{
	
}{}
\ifsiam{
	
}{}
\ifsiamsingle{
	
}{}
\ifmysiam{
	
}{}
\ifdcc{
	
}{}
\ifacm{
	
}{}
\ifspringerjournal{
	
}{}

\makeatother

\usepackage{hyperref}

\usepackage{comment}

\newcommand\Gstar{\ensuremath{G^*}}

\usepackage{mycommands}
\usepackage{benscommands}
	\usepackage{algorithm}%
\usepackage{algpseudocode} 

		\title{Simple approximation algorithms for Polyamorous Scheduling}
	\newcommand\email[1]{\texttt{#1}}
	\author{%
		Yuriy Biktairov%
		\footnote{University of Southern California, USA, \email{biktairo\,@\,usc.edu}}
	\and
		Leszek Gąsieniec%
		\footnote{University of Liverpool, UK, \email{\{lechu,\,n.namrata,\,b.m.smith,\,wild\}\,@\,liverpool.ac.uk}}
	\and
		Wanchote Po Jiamjitrak%
		\footnote{University of Helsinki, Finland, \email{wanchote.jiamjitrak\,@\,helsinki.fi}}
	\and
		Namrata\footnotemark[2]
	\and
		Benjamin Smith\footnotemark[2]
	\and
		Sebastian Wild%
			\footnote{University of Marburg, Germany, \email{wild\,@\,informatik.uni-marburg.de}\\[.5ex]
				Wanchote Po Jiamjitrak is supported by the Research Council of Finland grant No.\ 346968.\\
				Namrata and Sebastian Wild are supported by the Engineering and Physical Sciences Research Council of the UK grant EP/X039447/1.
			}
			\footnotemark[2]
	}
	
	\date{\small\today}

\begin{document}

\ifacm{}{\maketitle} %

\begin{abstract}
In Polyamorous Scheduling, we are given an edge-weighted graph and must find a periodic schedule of matchings in this graph which minimizes the maximal weighted waiting time between consecutive occurrences of the same edge. 
This NP-hard problem generalises Bamboo Garden Trimming and is motivated by the need to find schedules of pairwise meetings in a complex social group.
We present two different analyses of an approximation algorithm based on the Reduce-Fastest heuristic, from which we obtain first a 6-approximation and then a 5.24-approximation for Polyamorous Scheduling.
We also strengthen the extant proof that there is no polynomial-time $(1+\delta)$-approximation algorithm for the Optimisation Polyamorous Scheduling problem for any $\delta < \frac1{12}$ unless P = NP to the bipartite case.
The decision version of Polyamorous Scheduling has a notion of density, similar to that of Pinwheel Scheduling, where problems with density below the threshold are guaranteed to admit a schedule (cf.\ the recently proven 5/6 conjecture, Kawamura, STOC 2024). 
We establish the existence of a similar threshold for Polyamorous Scheduling and give the first non-trivial bounds on the poly density threshold.
\end{abstract}

\ifacm{%
	\maketitle%
}{}

%
%
%
%
%
%
%
%
%

%
%
%
%
%
%
%
%
%
%
%
%
%
%
%
%
%
%
%
%
%
%
%
%
%
%
%
%
%
%
%
%
%
%
%
%
%
%
%
%
%
%
%
%
%
%
%
%
%
%
%
%
%
%
%
%
%
%
%
%
%
%
%
%
%
%
%
%
%
%
%
%
%
%
%
%
%
%
%
%
%
%
%
%
%
%
%
%
%
%
%
%
%
%
%
%
%
%

%
%
%
%
%
%

%
%
%
%
%
%
%
%
%
%
%
%
%
%
%
%
%
%
%
%
%
%
%
%
%
%
%
%
%
%
%
%
%
%


\section{Introduction}
\label{sec:intro}
We prove that simple heuristics from Bamboo Garden Trimming can be generalized to Polyamorous Scheduling and that they yield a constant-factor approximation;  we complement this with hardness-of-approximation and structural results about the problem.

Polyamorous Scheduling is a recently introduced class of periodic scheduling problems.
\defi{Polyamorous Scheduling} problems each consider a ``polycule'', \ie{} a set of individuals connected by a set of pairwise relationships, each characterized by a value indicating its neediness, importance, or emotional weight. 
Their given objective is to devise a periodic schedule of pairwise meetings between couples that minimizes the maximum weighted waiting time between such meetings, subject to the constraint that each individual can meet with at most one person on any given day (see \wref{sec:preliminaries} for formal definitions).

Polyamorous Scheduling (Poly Scheduling in short) was first introduced by Gąsieniec, Smith and Wild in~\cite{GasieniecSmithWild2024}. 
We first consider the optimisation variant of Poly Scheduling namely, the \defi{Optimisation Polyamorous Scheduling} (OPS) problem where the objective is to find a schedule that minimizes the ``\emph{heat}'', \ie{} the worst pain of separation ever felt by any couple.
Unlike simpler periodic scheduling problems, such as the Bamboo Garden Trimming (BGT) problem~\cite{GasieniecEtAl2017}, Poly Scheduling is known to be hard to approximate: No poly-time $(\frac43-\varepsilon)$-approximation exists unless $P = \mathit{NP}$.

Prior to this work, the asymptotically best approximation ratio is achieved by the ``layering algorithm'' from~\cite{GasieniecSmithWild2024}, which gives a $3\lg(\Delta+1)$-approximation where $\Delta$ is the maximal number of partners of any person in the input.
We will show in this paper that this layering algorithm is best possible up to constant factors in its class of algorithms, namely all methods that work by scheduling relationships as series of \emph{disjoint} matchings, but 
that outside of this class, constant-factor approximations are possible.
Indeed, we show that a simple generalization of the \emph{Reduce-Fastest} strategy (with $x=2.89$) achieves a $5.24$-approximation.

We also consider \defi{Decision Polyamorous Scheduling} (DPS), where each couple has a maximum duration between meetings and a schedule must either be found or shown not to exist. Pinwheel Scheduling, the well-studied decision variant of Bamboo Garden Trimming, is a special case of DPS where the polycule is a star graph (similarly, BGT is the set of OPS polycules on star graphs). 
The recent proof of the long-standing $\frac{5}{6}$-density conjecture of Pinwheel Scheduling by~\cite{Kawamura2024} admits a strong sufficient criterion to certify feasibility for this class of polycules; prior to this work, this was lacking for Poly Scheduling.
More precisely, \cite{GasieniecSmithWild2024} presents the notion of ``poly density'' as a generalization of Pinwheel density, but the original definition based on an exponential-size linear program leaves open how to compute the poly density of an instance and hence doesn't provide an efficiently testable sufficient criterion.
In this paper, we show that poly density is always within a constant factor of the largest local Pinwheel density and can hence be well approximated.

A consequence of the above observation is that from the perspective of constant-factor approximations, any OPS instance's cost is dominated by a single BGT instance embedded in it.
This insight was the key ingredient to designing and analyzing a generalization of the Reduce-Fastest approximation~\cite{GasieniecEtAl2017} for Poly Scheduling.

A well-motivated special case of Poly Scheduling considers a \emph{bipartite} graph of relationships (Bipartite Poly Scheduling).
In the original motivation for scheduling meetings of polyamorous people, bipartite polycules correspond to a group of heterosexuals; it also models periodic scheduling scenarios with two explicit classes or hierarchies, such as the members of a sports team who each need to regularly train on a shared set of machines or with a shared set of trainers.
Apart from applications, there are strong indications in~\cite{GasieniecSmithWild2024} that bipartite instances may be more well behaved: All examples of infeasible decision polycules where showing infeasibility was nontrivial involve odd cycles. Moreover, the hardness-of-approximation proof in~\cite{GasieniecSmithWild2024} reduces the Chromatic Index problem to unweighted DPS instances; the latter is trivial on bipartite graphs (since they can always be $\Delta$-edge coloured).
We show that bipartiteness does indeed \emph{not} make the problem easier:  It remains NP-hard to find even approximate solutions to Bipartite Poly Scheduling.  As a consequence, the hardness of Poly Scheduling is intrinsically linked to the frequencies resp.\ weights of edges, not simply the embedded edge colouring problem.\

\subsection{Related Work}
\label{sec:further}
Poly scheduling problems belong to the broader class of periodic scheduling problems, many of which have attracted considerable attention lately~\cite{MR4470881, MR4640744, Kawamura2024, MR4387330}.

Arguably, the simplest periodic scheduling problem is~\defi{Pinwheel Scheduling}, which has been studied extensively since it was introduced in 1989~\cite{holte1989pinwheel}. In Pinwheel Scheduling we are given~$n$ positive integer frequencies~$f_1 \leq f_2 \leq \cdots \leq f_n$. 
The goal is to find a Pinwheel schedule, \ie{} an infinite sequence of tasks $1, \ldots, n$ such that any contiguous time window of length~$f_i$ contains at least one occurrence of~$i$ or to show that no such schedule exists. The \emph{density} of a Pinwheel Scheduling instance is $d := \sum_{i=1} ^{n} \frac{1}{f_i}$. It is easy to observe that $d \leq 1$ is a necessary condition for an instance to be schedulable, but this is not sufficient; in fact, there exists a threshold~$d^{\ast}$ such that~$d \leq d^{\ast}$ implies schedulability. 
A long sequence of works~\cite{holte1989pinwheel, MR1212158, MR1171436, MR1470043, chan1992general, fishburn2002pinwheel, ding2020branch, gkasieniec2022towards} successively improved bounds on~$d^{\ast}$, culminating in the very recent proof of the best possible bound of~$\frac{5}{6}$ by Kawamura’s~\cite{Kawamura2024}~-- resolving the long-standing conjecture of Chan and Chin from~1992~\cite{MR1171436}. Generalizations of Pinwheel Scheduling have also been studied, \eg, with jobs of different lengths~\cite{MR2226970, han1992scheduling}.

We will extensively use the natural optimisation variant of Pinwheel Scheduling, known as the \defi{Bamboo Garden Trimming} (BGT) problem, throughout the paper.  In BGT, a garden~$G$ of~$n$ bamboos~$b_1, b_2, \ldots, b_n$ with known respective daily growth rates~$g_1 \geq g_2 \geq \ldots g_n > 0$ is given. 
Initially, the height of each bamboo is set to zero. A gardener maintains the garden by trimming each bamboo to height zero according to some schedule. The objective of BGT is to design such a schedule, minimizing the maximum observed height of the tallest ever bamboo, subject to the constraint that the gardener can cut only one bamboo at the end of each day and is not allowed to attend the garden at any other time. This problem was introduced by Gąsieniec, Kalsing, Levcopoulos, Lingas, Min, and Radzik in 2017~\cite{GasieniecEtAl2017} and has since attracted a lot of attention.  
This original paper proved that there exists a natural lower bound of~$H:= g_1 + g_2 + \cdots g_n$ on the maximum height of a bamboo. 
One particularly popular line of research into BGT concerns constant-factor approximation ratios for offline scheduling. Gąsieniec et al ~\cite{GasieniecEtAl2017} introduced a 2 approximation, but this has since been improved to the current best of $\frac{4}{3}$ by Kawamura~\cite{Kawamura2024}.
For online scheduling, various greedy approaches were taken. Some of the most natural strategies are to trim the highest bamboo, aka, ~\defi{Reduce-Max}~\cite{alshamrani2015reduce}, to trim the fastest growing bamboo whose height is above a certain threshold, aka,~\defi{Reduce-Fastest}~\cite{GasieniecEtAl2017}, and to trim the bamboo above a certain threshold that will soonest achieve the desired height, aka,~\defi{Deadline-Driven}~\cite{kuszmaul2022bamboo}. 

The first approximation ratio to be proved for Reduce-Max was ~$\mathcal{O}(\log H)$ ~\cite{GasieniecEtAl2017}, but after extensive computer experiments~\cite{MR3945421}, it was conjectured that Reduce-Max keeps the maximum bamboo height within~$\mathcal{O}(H)$. In fact, the authors conjectured a much stronger bound of $2H$ for Reduce-Max. However, this cannot be arbitrarily low. The authors of ~\cite{gkasieniec2024perpetual} constructed an input BGT instance using which they proved that the approximation ratio of Reduce-Max cannot be less than~9/8. On the other hand, Bil\`o et al.~\cite{MR4387330} proved a bound of $9H$ on the maximum bamboo height under the Reduce-Max algorithm, which was later improved by Kuszmaul~\cite{kuszmaul2022bamboo} to $4H$ -- the best bound to be proved so far.

Another popular algorithm is Reduce-Fastest($x$), in which the fastest-growing bamboo with a height of at least~$xH$ is trimmed each day.
The bounds known for Reduce-Fastest are a bit better than that of Reduce-Max. The first approximation ratio of 4 was proved for Reduce-Fastest(2) in~\cite{GasieniecEtAl2017}. The same paper~\cite{MR3945421} that conjectured the $2H$ bound of Reduce-Max, also conjectured that Reduce-Fastest(2) and Reduce-Fastest(1) achieve the maximum height of $3H$ and $2H$, respectively. Subsequently, Bil\`o et al.~\cite{MR4387330} proved that for~$x = 1 + 1/\sqrt{5} \approx 1.45$, the maximum bamboo height achieved  is within the approximation ratio of~$\frac{(3+\sqrt{5})}{2} \approx 2.62$. Kuszmaul~\cite{kuszmaul2022bamboo} showed recently that for any~$x \geq 2$, Reduce-Fastest($x$) keeps all bamboos below~$(x+1)H$, thereby proving one of the conjectures from earlier~\cite{MR3945421}. Kuszmaul also proved that no matter the value of~$x$, Reduce-Fastest($x$) does not give a bound better than $(2.01)H$, thus disproving the conjecture that Reduce-Fastest($1$) keeps all bamboos below~$2H$.

Kuszmaul~\cite{kuszmaul2022bamboo} also proved an approximation ratio of 2 for the simple Deadline-Driven algorithm. The algorithm was first introduced by Liu and Layland~\cite{MR0343898} in the context of a related scheduling problem in the early 1970s.

The Bamboo Garden Trimming problem is closely related to \emph{Cup Games}, which share the mechanic of growth and cutting and are used to model similar applications~\cite{backlog_multi_cup_games, vanilla_cup_games}. Single processor cup games follow a similar daily routine in which an adversarial filler distributes 1 unit of water freely between a set of cups, then an emptier selects a cup and removes up to 1 unit of water from it. Multi-processor versions are also of interest, relating to the multi-gardener version of Bamboo Garden Trimming~\cite{kuszmaul2022bamboo}. As with Bamboo Garden Trimming, greedy algorithms perform well.

Further periodic scheduling problems with less direct connections to Poly Scheduling have also been studied.  
\emph{Patrolling problems} typically involve periodic schedules: for example,~\cite{MR0343898} finds schedules for a fleet of $n$ identical robots which patrol points in a metric space; 
the \emph{Continuous BGT Problem}~\cite{gkasieniec2024perpetual} sends a single robot to points with different frequencies requirements; 
the \emph{Point Patrolling Problem} studied in ~\cite{MR4143586}  assigns one of $n$ workers a single, daily recurring task subject to the constraint that each worker $i$ requires a break of $a_i$ days between work days; 
the \emph{Replenishment Problems with Fixed Turnover Times}, studied in~\cite{MR4467817}, considers vertices in a graph which must be visited with given frequencies, constrained by the goal of minimizing the length of the tour used to visit them.

\subsection{Our Results}
\label{sec:our_results}
Despite being an interesting combinatorial optimisation problem, Poly Scheduling has only been formally studied in~\cite{GasieniecSmithWild2024}. 
This paper introduced several approximations, the notion of poly density, and the first inapproximability result for any periodic scheduling problem; we introduce better approximations, bound and exploit poly density, and extend this inapproximability result to the bipartite case.

An OPS instance~$\mathcal{P}_o = (P, R, g)$ consists of an undirected graph~$(P, R)$ of people $P$ connected by relationships $R$,
along with a desire growth rate~$g_e : R \mapsto \mathbb{R}_{>0}$ for each relationship $e\in R$. 
The \emph{current heat} of an edge is the total desire grown since the last time it was scheduled;
every day, the scheduler must choose a set of all-day meetings such that each person $v\in P$ is not double-booked, 
with the goal of minimising the highest heat ever experienced.

In \wref{sec:approxes}, we define a simple generalization of the Reduce-Fastest($x$) algorithm introduced by \cite{GasieniecEtAl2017}:
Broadly speaking, Reduce-Fastest$(x)$ is an online scheduling method which chooses a matching to schedule each day;
the matching with the fastest-growing edges whose current heat is above some user-selected threshold~$x$.
In \wref{sec:5.24-approx}, we introduce a simple analysis that shows that Reduce-Fastest$(4)$ gives a 6-approximation, while more in-depth analysis shows that Reduce-Fastest$(2 + \frac{2\sqrt{5}}{5} \approx 2.89)$ gives a $3+\sqrt{5} \approx 5.24$ approximation.
These both substantially improve on the~$O(\log \Delta)$ approximation ratio in~\cite{GasieniecSmithWild2024}.
Formal definitions of OPS and Reduce-Fastest$(x)$ for OPS are introduced in \wref{sec:preliminaries}.

When looking for exact solutions for certain OPS instances, it is well established that schedules with exponentially long periods are inevitable, and whether any feasible instance of Pinwheel Scheduling admits a succinctly encodable schedule remains an open problem.
We show that the same is not true for \emph{approximate} solutions. In \wref{sec:length}, we show that for any arbitrary OPS schedule, we can construct a new periodic schedule with polynomial length and a maximum heat of at most four times that of the original schedule; we thus never have to deal with superpolynomial schedules for constant-factor approximations.

\wref{sec:ops-density} revisits the idea of poly density introduced by \cite{GasieniecSmithWild2024} and shows that it is narrowly sandwiched by the maximum personal growth rate of any person in the polycule:
$\Gstar = \max_{v\in P} \sum_{e\in R : v\in e} g(e)$.

\begin{theorem}[Poly Density Approximation]
    \label{thm:OPS_density_combined}
    For an OPS instance with poly density $\Bar{h}^*$ and maximum personal growth rate $\Gstar$, we have $\Gstar \leq \Bar{h}^* < \frac32 \Gstar$.
\end{theorem}

We then turn our attention to the poly density of Decision Polyamorous Scheduling. 
Each instance of DPS, or \emph{DPS polycule} $\mathcal{P}_d = (P,R,f)$ consists of a set $P$ of people, a set $R$ of relationships, and a set $f$ of “meetup frequencies” $f_e$. The goal is to either find a schedule of pairwise meetings such that each couple $e=\{i, j\}$ meets at least every $f_e$ days or to report that no such schedule exists (subject to the constraint that each person can attend at most one meeting per day). 
A DPS polycule can naturally be represented as a graph of people with the edges representing their relationships; since each person can attend at most one meeting per day, the edges scheduled on any given day must form a matching in this graph.

\wref{sec:dps-density} shows that the poly density of DPS generalizes the density of Pinwheel Scheduling and establishes a poly-density threshold.

\begin{theorem}[Density Threshold]
\label{thm:densityThreshold}
    Any  DPS instance $\mathcal{P}_d = (P, R, f)$ with poly density $d \le \frac14$ is schedulable.
\end{theorem}

Finally, in \wref{sec:inapproximability} we strengthen the inapproximability result introduced by \cite[Theorem 1.3]{GasieniecSmithWild2024} by proving that it holds even in the bipartite case:

\begin{theorem}[SAT Hardness of approximation]
\label{thm:sat-inapprox}
	Unless P = NP, there is no polynomial-time $(1+\delta)$-approximation algorithm 
	for the bipartite Optimisation Poly Scheduling problem for any $\delta < \frac1{12}$.
\end{theorem}

We prove this by constricting a bipartite DPS polycule $\mathcal P_{d\varphi}$ from an arbitrary 3SAT instance $\varphi$ such that $\mathcal P_{d\varphi}$ is schedulable iff $\varphi$ is satisfiable.

\section{Preliminaries}
\label{sec:preliminaries}
This section introduces some general notations which will be used throughout the paper.

We write~$[m,n]$ for $\{m, m+ 1, \ldots, n\}$ and~$[n]$ for ~$[1,n]$. For a set~$\mathcal{A}$, we denote its powerset by~$2^{\mathcal{A}}$. All graphs~$(V, E)$ are simple and undirected, and the \emph{maximum} degree in~$(V, E)$ is denoted by~$\Delta(V,E)$. An edge colouring of a graph is an assignment of colours to its edges so that no two adjacent edges are assigned the same colour. The chromatic index of a graph~$(V, E)$, denoted~$\chi_{1}(V,E)$, is the smallest number of colours with which~$(V, E)$ can be edge coloured. Vizing~\cite{MR0180505} has shown that for any graph~$(V, E)$, the chromatic index~$\chi_{1}(V,E)$ is either its maximum degree~$\Delta$, or else is $\Delta + 1$.
We represent the set of vertices in the neighbourhood of~$v \in V$ by~$\mathcal{N}(v)$.

We now give formal definitions for the optimisation and decision versions of the Polyamorous Scheduling problem, as introduced by~\cite{GasieniecSmithWild2024}.
For a set of edges $R$,
an (infinite) \defi{schedule} is a function assigning days to subsets of edges: $S: \mathbb{N}_0 \rightarrow 2^R$. 
Given a schedule $S$ and an edge $e \in R$, we define the \defi{recurrence time}~$r(e)=r_S(e)$ of $e$ in $S$ as the maximal time between consecutive occurrences of $e$ in $S$, formally
\begin{align}
    r_S(e) \wrel{:=} \sup_{d \in \mathbb{N}} \,
    \begin{cases}
        d+1 &\exists t \in \mathbb{N}_0: e \notin S(t) \cup S(t+1) \cup \cdots \cup S(t+d-1); \\
        1 &\text{otherwise.}
    \end{cases}
\end{align}
(We can allow $r_S(e) = \infty$ if $S$ leaves growing gaps between consecutive occurrences of $e$, but all schedules in the following will have finite recurrence times for all edges.)

\begin{definition}[Optimisation Polyamorous Scheduling]
An OPS instance~$\mathcal{P}_o = (P, R, g)$ consists of an undirected graph~$(P, R)$, where the vertices~$P :=  \{p_1, \ldots, p_n\}$ are $n$ persons and the edges~$R$ are pairwise relationships, along with a desire growth rate~$g : R \to \mathbb{R}_{>0}$ for each relationship in~$R$. 
A schedule~$S : \mathbb{N}_0 \to 2^{R}$ is valid if, for all days,~$t\in \mathbb{N}_0$, $S(t)$ is a matching in $P_o$.
The goal is to find a valid schedule that minimizes the heat~$h := h(S) := \max_{e\in R} \, g(e)\cdot r_S(e)$.
\end{definition}

\begin{definition}[Decision Polyamorous Scheduling]
\label{def:DPS}
A DPS instance~$\mathcal{P}_d = (P, R, f)$ consists of an undirected graph~$(P, R)$ with integer frequencies $f: R \to \mathbb{N}$ for each relationship.

The goal is to find a schedule $S: \mathbb{N}_0 \rightarrow 2^R$, such that
\begin{enumerate}[label=(\alph*)]
\item (no conflicts) for all days $t \in \mathbb{N}_0, S(t)$ is a matching in $\mathcal{P}_d$, and
\item (frequencies) for all $e \in R$, we have $r_S(e) \le f(e)$,
\end{enumerate}
or to report that no such schedule exists. In the latter case, $\mathcal{P}_d$ is called infeasible.
\end{definition}
We write $f_e$ as a shorthand for $f(\{p_i, 
p_j\})$ resp. $f(e)$. 
It is easy to show that an infinite schedule exists if and only if a periodic schedule exists, \ie{} a schedule where there is a $T \in \mathbb{N}$ such that for all $t$, we have $S(t)=S(t+T)$.

It will be convenient for much of the following to use \defi{Normalised} OPS instances, which we define in the following way:
Considering an OPS polycule $\mathcal{P}_o = (P, R, g)$, 
define the personal growth rate $G_v$ for person $v\in P$ as~$G_v := \sum_{e \in R:v \in e} g(e)$ 
and 
the maximum personal growth rate across all $v\in P$ as $\Gstar := \max_{v \in P} G_v$.
$\mathcal{P}_o$ is normalised, or in normal form, if $\Gstar = 1$. 
Any OPS polycule $\mathcal{P}_o$ can be trivially normalised by dividing each growth rate $g_e\in g$ by $\Gstar$, and restored by inverting this.

The following pair of lemmas are taken from~\cite{GasieniecSmithWild2024}; the first shows how to reduce OPS to DPS, while the second shows how DPS solves OPS.

\begin{lemma}[OPS to DPS, \cite{GasieniecSmithWild2024}]
\label{lem:ops-dps}
For every combination of OPS instance $\mathcal{P}_o=(P, R, g)$ and heat value $h$, there exists a DPS instance $\mathcal{P}_d=(P, R, f)$ such that
\begin{enumerate}[label=(\alph*)]
\item any feasible schedule $S: \mathbb{N}_0 \rightarrow 2^R$ for $\mathcal{P}_d$ is a schedule for $\mathcal{P}_o$ with heat $\leq h$, and
\item any schedule $S^{\prime}$ for $\mathcal{P}_o$ with heat $h^{\prime}>h$ is not feasible for $\mathcal{P}_d$.  
\end{enumerate}
\end{lemma}

\begin{lemma}[DPS to OPS, \cite{GasieniecSmithWild2024}]
\label{lem:dps-ops}
Let $\mathcal{P}_d=(P, R, f)$ be a DPS instance. Set $F:=\max _{e \in R} f(e)$. There is an OPS instance $\mathcal{P}_o=(P, R, g)$ such that the following holds.
\begin{enumerate}[label=(\alph*)]
\item If $\mathcal{P}_d$ is feasible, then $\mathcal{P}_o$ admits a schedule of heat $h \leq 1$.
\item If $\mathcal{P}_d$ is infeasible, then the optimal heat $h^*$ of $\mathcal{P}_o$ satisfies $h^* \geq(F+1) / F$.
\end{enumerate}
\end{lemma}


\section{Approximating Optimisation Polyamorous Scheduling}
\label{sec:approxes}
This section will examine approximation algorithms for OPS in 3 ways: \wref{sec:lower-bounds} considers several lower bounds on OPS, \wref{sec:reduce-fastest} evaluates reduce-fastest as an approximation algorithm for OPS, and \wref{sec:length} introduces a method for constructing polynomial schedules for OPS and BGT.

As OPS generalises BGT, with instances on star-graphs being fully equivalent, these problems share both features and algorithms. Particularly of interest here is the \emph{Reduce-Fastest} heuristic, which we generalise here and will use in multiple parts of this section:

\begin{definition}[Reduce-Fastest$(x)$ for OPS]
\label{def:reduce-fastest}
    Begin with an OPS instance~$\mathcal{P}_o = (P, R, g)$ and relabel to 
    sort edges $e\in R$ by decreasing growth rates $g(e)$, breaking ties arbitrarily.

    Construct an online schedule using the fixed parameter~$x$ by  scheduling a greedy matching $M$ each day, built the following way:

    For all edges $e=\{u,v\}\in R$ in order of decreasing growth rate, add $e$ to $M$ if and only if
    \begin{itemize}
        \item the current heat of $e$ is at least ~$x\cdot \Gstar$, and
        \item no edge incident at $u$ or $v$ is already in $M$.
    \end{itemize}
\end{definition}

Note that on a star graph, each matching $M$ will contain a single edge, so our OPS Reduce-Fastest behaves like the BGT Reduce-Fastest heuristic introduced by \cite{GasieniecEtAl2017} on such instances.

\subsection{Lower Bounds}
\label{sec:lower-bounds}
We will begin by defining Disjoint Matching algorithms and demonstrating them to be a dead end in \wref{sec:disjoint_matching_algs}, then show a simple but general lower bound from the local BGT instance in \wref{sec:lb-BGT}, then a lower bound for the Reduce-Fastest$(x)$ algorithm in \wref{sec:lb-rf}.

\subsubsection{Disjoint Matching Algorithms}
\label{sec:disjoint_matching_algs}

In addition to introducing OPS, \cite{gkasieniec2024perpetual} introduced multiple proofs that OPS is NP-hard, several lower bounds for it, and the \emph{Layering Algorithm} -- an $\mathcal{O}(\log \Delta)$-approximation for OPS.

This algorithm partitions edges into layers based on their growth rates, uses a fixed round-robin schedule for each layer, and then interleaves these schedules into another fixed round-robin schedule.
Further details are not needed here, only that each edge only appears in a single matching, and thus, all matchings in schedules produced by the Layering Algorithm are \emph{disjoint}. 

We define Disjoint Matching algorithms as algorithms whose schedules contain only disjoint matchings, and will show that such algorithms obey the following:

\begin{theorem}
    \label{thm:lb-DMA}
    The approximation ratio of any Disjoint Matching algorithm for OPS is $\Omega(\log n)$. 
\end{theorem}
\begin{proof}
    Consider an OPS instance~$\mathcal{P}^*_o = (P, R, g)$ consisting of $d$ disjoint stars $*_1,\ldots,*_i,\ldots, *_d$, such that star $i$ has $i$ edges and $g= \frac{1}{i}$ for all edges in star $i$. The first four such stars are shown in \wref{fig:disjoint-stars}. 
    The optimal heat for $\mathcal{P}^*_o$ is 1, and can be obtained by scheduling each star independently in a round-robin fashion. 

    Any set of disjoint matchings $\mathcal{M}$ which covers all edges $e\in R$ will consist of at least $d$ matchings due to the largest star $*_d$. 
    Let~$g_M^{*}$ be the maximum growth rate of any edge in matching~$M_i\in \mathcal{M}$, \ie{} $g_i^* := \max_{e\in M_i} g(e)$ and $g^* = \max_{i\in\mathcal{M}} g_i^*$. 
    In the given disjoint stars, there is at least one matching with $g_i^* = 1$, at least two matchings with $g_i^* \ge 1/2$, at least three with $g_i^* \ge 1/3$, and so on. 
    Once we commit to scheduling the matchings as atomic objects, the problem has become equivalent to a BGT instance with~$d$ bamboos of growth rates~$g_1^{*},\ldots,g_d^*$. 
    The optimal height for any schedule of this instance is at least~$\sum_{i=1}^d g_i^* \ge \sum_{i=1}^{\Delta} 1/i \sim \ln \Delta$.
    Since $n = \sum_{i=1}^\Delta 2i-1 = \Delta^2$, any such schedule for $\mathcal{P}^*_o$ thus has heat $\Omega(\log\Delta)=\Omega(\log n)$, compared to the optimal heat of $1$.
\end{proof}

This means that, in its ``league'' of algorithms that first divide the graph into a set of disjoint matchings and then schedule these matchings as atomic objects, the Layering Algorithm's approximation ratio is best possible up to constant factors. To do much better than this, different kinds of approach will be needed

\begin{figure}[hbt]
\centering
\includegraphics{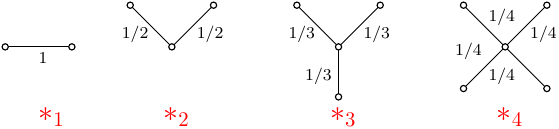}
\caption{An instance of disjoint set of four stars that give the worst-case lower bound of 
$\Omega(\log n)$ for the maximum heat for an OPS instance when disjoint matching is scheduled.}
\label{fig:disjoint-stars}
\end{figure}

\subsubsection{Lower bound due to the local BGT instance}
\label{sec:lb-BGT}
Recall the BGT problem introduced in \wref{sec:further}, and that BGT is the special case of OPS on star graphs. Thus, the largest lower bound obtained from the embedded BGT instances in the OPS instance will be a natural lower bound for OPS.
More formally, consider a vertex~$v \in P$, recalling that the personal growth rate $G_v$ for person $v\in P$ as~$G_v := \sum_{e \in R:v \in e} g(e)$ and the maximum personal growth rate across all $v\in P$ as $\Gstar := \max_{v \in P} G_v$.
We claim that~$\Gstar$ is a natural lower bound on the maximum heat of any schedule of an OPS instance. 
\begin{theorem}
\label{thm:bgtlowerbound}
For any OPS instance $\mathcal{P}_o = (P, R, g)$, $\Gstar$ is a lower bound on the maximum heat of any schedule for $\mathcal{P}_o$.
\end{theorem}

\begin{proof}
This can be easily proved by contradiction. Let $v^* = \argmax_{v \in P}  \sum_{e:v \in e}g(e)$ be the vertex that corresponds to~$\Gstar$. Assume that there exists a periodic schedule that keeps the heat of all edges incident on~$v^*$ to some heat~$h < \Gstar$. 
In the neighbourhood of $v^*$, $\mathcal{N}(v^*)$, we can only schedule one edge on any given day.
After~$\bigl\lfloor\frac{|N(v^*)|\cdot h} {\Gstar-h}\bigr\rfloor+1$ days, the heat of at least one edge must be greater than~$h$, a contradiction.
\end{proof}

\subsubsection{Lower bound for Reduce-Fastest} 
\label{sec:lb-rf}
The next section will analyse the approximation ratio of the Reduce-Fastest heuristic, but first, we will demonstrate a lower bound for this heuristic. The proof concerns a fully connected OPS polycule which is nominally fair (all growth rates are the same), but which is labelled adversarially with respect to a single unfortunate individual.

\begin{theorem}
\label{thm:rf-lower-bound}  
    When applied to normalised OPS instances $\mathcal{P}_o = (P, R, g)$, Reduce-Fastest$(x)$ with $x\geq 2$ can produce schedules with heat $h\geq x+2-\frac{1}{|P|-1}$.
\end{theorem}

\begin{proof}
For some odd~$n$, consider a normalised OPS instance $\mathcal{P}_o = (P, R, g)$ where $(P,R)$ is the complete graph $K_{n}$ on $n$ vertices, and has uniform growth rates of $\frac{1}{n-1}$.
Reduce-Fastest($x$) will not schedule edges until they reach the threshold heat $x$, and will then schedule edges according to growth rate (all tied), then to an arbitrary choice of their ordering to break ties. To construct an adversarial ordering of edges, we will consider two subgraphs: the arbitrary induced subgraph~$K_{n-1}$ of~$P$, and the subgraph of edges in $K_{n}$, but not $K_{n-1}$, called $K_{n}-R(K_{n-1})$.

Since~$n$ is odd,~$n-1$ is even, and $K_{n-1}$ is 1-factorizable, \ie{} edges of~$K_{n-1}$ can be decomposed into~$n-2$ disjoint matchings, each covering all vertices in~$K_{n-1}$. As only one vertex appears in $K_{n}$ but not $K_{n-1}$, $K_{n}-R(K_{n-1})$ requires a further $n-1$ matchings to cover, as each matching can only contain a single edge from $K_{n}-R(K_{n-1})$.

So, for some initial ordering of edges $e\in R$, Reduce-Fastest$(x)$ will first wait for all edges to grow to heat $x$, then schedule matchings in $K_{n-1}$, then matchings covering $K_{n}-R(K_{n-1})$. This means that the last edge is scheduled $(n-2+n-1) = 2n-3$ days later. This last edge thus starts with heat at least $x$ and grows a further heat of~${(2n - 3)}\cdot \frac{1}{n-1} = 2 - \frac{1}{n-1}$ before being scheduled.

Note that in this $2n-3$ day period, the heat of each edge grows by~$\frac{1}{n-1}$ per day, so no previously scheduled edge will reach the threshold of $x \ge 2$, and thus no edge will interfere with the scheduling sequence.
\end{proof}

\begin{remark}
\label{rem:threshold}
Note that if Reduce-Fastest is applied to the above graph with~$x<2$, the maximum heat will tend to infinity for an appropriately large $n$. This is because the edges which are scheduled first will cross the threshold $x$ for a second time before the last edge is scheduled, causing this last edge to grow indefinitely.
\end{remark}

\subsection{Reduce-Fastest}
\label{sec:reduce-fastest}
This section gives two analyses for the approximation ratio of the OPS version of Reduce-Fastest~$(x)$, defined by \wref{def:reduce-fastest}. 

Both analyses use a notion of \emph{blocking}: Each day, Reduce-Fastest constructs a matching $M$ by greedily selecting edges with heat at least $\Gstar x$ according to their growth rates (with ties broken by some initial ordering). This means that for some edge $e$ with heat $h \ge \Gstar x$ not to be added to $M$, a higher priority edge containing one of the vertices in $e$ must block $e$ by being added to $M$ before $e$ was considered.

\subsubsection{6-Approximation Algorithm}

In this section, we prove that Reduce-Fastest$(4)$ is a 6-approximation for OPS by proving \wref{thm:6-approx-ratio}.
This analysis is inspired by the similar proof of an $(x+1)$-approximation algorithm for the BGT problem in~\cite{kuszmaul2022bamboo}. 

\begin{theorem}
\label{thm:6-approx-ratio}
Given a normalised OPS instance $\mathcal{P}_o = (P, R, g)$; for all $x \ge 4$, the maximum heat achieved by an edge in Reduce-Fastest$(x)$ is strictly less than $x + 2$.
\end{theorem}

\begin{proof}
Suppose for the sake of contradiction that some edge $e_i \in R$ achieves heat of at least $x + 2$ at time $t_2$ after most recently reaching heat $x$ at time $t_1$. This implies that on every day in the interval $[t_1, t_2)$,
~$e_i$ was blocked by a higher-priority adjacent edge. 
Consider the sequence of these $t_2 - t_1$ blocking edges and denote the set of distinct edges in this sequence by $S$.
$e_i$ could be blocked by two separate higher-priority edges at the same time -- in this case, arbitrarily select one to include in the sequence
For all $e_j \in S$, we denote by $m_j$ the number of times that $e_j$ appears in the sequence.
For an edge to appear $m_j$ times in our sequence, its growth rate must follow:

\begin{claim}
\label{cla:inequality-growth-rates}
For all $e_j \in S$, we have $g(e_j) > m_j \cdot g(e_i)$.
\end{claim}

\begin{proof}
Begin by considering the case of $m_j = 1$, which simply implies that $g(e_j) \ge g(e_i)$. This follows from the definition of Reduce-Fastest($x$) (\wref{def:reduce-fastest}) -- the edge with faster heat growth rate is scheduled earlier and hence blocks the edge~$e_i$.

Given that edge $e_j$ can have any initial heat, it could be scheduled as early as $t_1$. After this, it must grow to heat $x$ before each time it is scheduled again, $m_j - 1$ times total. Thus, to achieve $m_j \ge 2$, edge $e_j$ has to grow a total heat of at least $x(m_j - 1)$ in the interval $(t_1, t_2)$. In the same time interval $\mathcal{T}$, $e_i$ must grow less than 2 units of heat, as it starts with heat $x$ and doesn't reach a heat of $x + 2$ until $t_2$. Thus 
\[\frac{2}{g(e_i)} \wrel> \mathcal{T} \wrel\ge \frac{x(m_j-1)}{g(e_j)}\]
and
\[g(e_j) \wrel> x(m_j - 1)\frac{g(e_i)}{2} \wrel> 2(m_j - 1) g(e_i) \wrel> m_j g(e_i),\]
for $ x\ge 4$ and $m_j \ge 2$.
This completes the proof of the claim.
\end{proof}

Using \wref{cla:inequality-growth-rates}, we can argue from the number of blockers:
\[
t_2 - t_1 
\wrel\le \sum_{e_j \in S} m_j 
\wrel< \sum_{e_j \in S} \frac{g(e_j)}{g(e_i)} 
\wrel= \frac{1}{g(e_i)} \sum_{e_j \in S} g(e_j) 
\wrel\le \frac{2}{g(e_i)} (1 - g(e_i)),
\]
where the last inequality is a product of normalisation -- since~$\Gstar = 1$, for each edge~$e_i = \{u, v\}$, we have~$G_u \le 1$ and~$G_v \le 1$.

Since the heat of $e_i$ is less than $x + g(e_i)$ at $t_1$, the heat level reached by $t_2$ is strictly less than
\[
x + g(e_i) + (t_2 - t_1) g(e_i) 
\wwrel\le x + g(e_i) \biggl(1 + \frac{2}{g(e_i)} \bigl(1 - g(e_i)\bigr)\biggr) 
\wwrel= x + 2 - g(e_i) 
\wwrel< x + 2, 
\]
which is a contradiction.
\end{proof}

\begin{corollary}
Reduce-Fastest$(4)$ is a 6-approximation of OPS.
\end{corollary}

From Theorems~\ref{thm:rf-lower-bound} and~\ref{thm:6-approx-ratio}, the following simple corollary follows.

\begin{corollary}
\label{cor:tight-bound}
Given an OPS instance $\mathcal{P}_o = (P, R, g)$, 
the $x + 2$ bound on the maximum heat achieved by Reduce-Fastest$(x)$ is tight for~$x \ge 4$ and large $n$. 
\end{corollary}

\subsubsection{5.24-Approximation Algorithm}
\label{sec:5.24-approx}
In this section, we prove our best bound on Reduce-Fastest: Reduce-Fastest(2.89) is a 5.24 approximation for OPS. 
The analysis follows the related proof of approximation algorithm for the BGT problem
in~\cite{MR4387330}.

\begin{theorem}
\label{thm:5.24-approx-ratio}
    Given a normalised OPS instance $\mathcal{P}_o = (P, R, g)$; for all $x > 2$, the maximum heat achieved by an edge in Reduce-Fastest$(x)$ is bounded by
    \begin{align*}
        \max \left\{x + \frac{x^2}{4(x - 2)},\;
        x + \frac{1}{2} + \frac{x^2}{4(x - 1)}  \right\}.
    \end{align*}
\end{theorem}

\begin{proof}
Let $S_{\mathrm{rf}}$ be the schedule produced when Reduce-Fastest$(x)$ is applied to~$\mathcal{P}_o$, $h=h(S_{\mathrm{rf}})$ be the heat of this schedule, and~$e_i = \{u,v\}$ be an edge which reaches heat~$h)$ on day~$T$. 
Before reaching day~$T$, the heat of~$e_j$ has grown for some period without being scheduled, reaching threshold~$x$ on day 0 in this period.
As~$e_i$ was not scheduled in the interval $[0, T - 1]$, for each day in this interval, there must have been some higher-priority edge adjacent to~$e_i$ which was scheduled instead;
construct the \emph{blocking sequence}~$S$ from these $T$ blocking edges.
On any day which has two separate blocking edges, arbitrarily choose one to add to~$S$ and add the other to the \emph{suppressed sequence}~$Q$.
Let $N$ be the number of distinct edges in~$S$.

Let the \emph{volume} $V$ be the overall growth of $e_i$ and its adjacent edges in the days of the interval $[0, T - 1]$. 
Because of normalisation,~$G_u \leq 1$ and~$G_v \le 1$, so the heat of this set of edges grows by at most $2 - g(e_i)$ per day, and $V \le T(2 - g(e_i))$. 

For each edge~$e_j\in S$, call the first time it was scheduled in $S$ the \emph{first-time scheduling} and further times \emph{repeated schedulings}. 
The volume of heat associated with edge $e_j$ is either removed by the first-time scheduling, removed by repeated scheduling, or is \emph{leftover volume} -- heat that remains after day $T-1$.

In order to bound the amount of volume $V_r$ removed by repeated schedulings, we first calculate the volume $V_f$ that is accumulated before the first time each edge is scheduled in $S$, then the volume $V_l$ which is leftover after the last time each edge is scheduled in $S$.
Notice that for each edge $e_j\in S$, it holds that $g(e_j) \ge g(e_i)$. 
If $e_j$ appears in $S$ for the first time on day $d$, then the volume removed on day $d$ in addition to the total volume removed by the preceding schedulings of $e_j$ in~$Q$ will be at least~$(d + 1)g(e_j) \ge (d + 1)g(e_i)$. 
This is because~$(d + 1)g(e_j)$ is the total growth of $e_j$ up to and including the  day~$d$, all of which will be removed when~$e_j$ is scheduled. 
This implies that the amount of volume removed by all first-time schedulings in addition to the preceding schedulings of the corresponding edges in the suppressed sequence~$Q$ will be given by $V_f\geq\sum_{j = 1}^N j g(e_i) = \frac{N(N+1)}{2} \cdot g(e_i)$, since exactly one edge is added to the sequence~$S$ per day. 

Moreover, if $e_j$ appears in the sequence~$S$ for its last time at day $T - 1 - d$, then the leftover volume of~$e_j$ in addition to the volume removed by the following schedulings of~$e_j$ in~$Q$ is at least $d' \cdot g(e_i)$. Therefore, for all $N$ edges, the total volume left after the last scheduling of $e_j$ in $S$ will be $V_l = \sum_{j = 1}^N (j - 1)  g(e_i) = \frac{N(N-1)}{2} \cdot g(e_i)$ . Finally, the edge $e_i$ is not scheduled in the interval $[0, T - 1]$ but does grow by exactly $T g(e_i)$ during this time.

We can then bound the total volume removed by repeated schedulings of edges in~$S$ as
\begin{align*}
V_r &\wwrel\le V - \left( \frac{N(N + 1)}{2} + \frac{N(N - 1)}{2} \right) g(e_i) - T g(e_i) \\
&\wwrel= V - N^2 \cdot g(e_i) - T g(e_i) \\
&\wwrel\le T(2 - g(e_i)) - N^2 g(e_i) - T g(e_i) \\
&\wwrel= 2T(1 - g(e_i)) - N^2 g(e_i). 
\end{align*}

Each of these removes heat at least $x$, so the number
$N_r$ of repeated schedulings is upper bounded by
$\frac{V_r}{x} = \frac{2T(1 - g(e_i)) - N^2 g(e_i)}{x}$, which gives us the
following bound on the total number of elements in the sequence
\[
|S| \wwrel= T \wwrel= N + N_r \wwrel\le N + \frac{2T(1 - g(e_i)) - N^2 g(e_i)}{x},
\]
which, with $x \ge 2$, implies
\[
T \wwrel\le \frac{Nx - N^2 g(e_i)}{2g(e_i) + x - 2} 
\]

Consider the minimum value of the right-hand side; as~$g(e_i)$ and~$x$ are fixed, we can view it as a function of $N$: $R(N) = (Nx - N^2 g(e_i))/({2g(e_i) + x - 2})$. 
The function~$R(N)$ is a concave downward parabola that attains its maximum at~$N = \frac{x}{2 g(e_i)}$, giving us 
\[
T \wwrel\le  R\left(\frac{x}{2 g(e_i)}\right) \wwrel= \frac{x^2}{4g(e_i)(2g(e_i) + x - 2)} 
\]

Using this upper bound to $T$, we now bound the overall growth of the edge
$e_i$, \ie{} the maximum heat $h(S_{\mathrm{rf}})$. 
At day $d=0$, $e_i$ has heat at most $x+g(e_i)$ by our choice of the time interval, and in the next $T$ days it grows by
$T g(e_i)$. Hence:
\[
h(S_{\mathrm{rf}}) \wwrel\le x + g(e_i) + T g(e_i) \wwrel\le x + g(e_i) + \frac{x^2}{4(2g(e_i) + x - 2)}.
\]
We now show that the claimed bound follows.
We distinguish two cases; first assume $g(e_i) \le \frac12$.
Viewed as a function of $g(e_i) \in [0,\frac12]$, the upper bound on $h(S_{\mathrm{rf}})$ is a convex function
since its second derivative w.r.t. $g(e_i)$ is
\[
	\frac{2x^2}{(2g(e_i)+x-2)^3} 
	\wrel > 0 
\]
(for $x> 2$).
The maximum (worst upper bound) is thus attained at the extreme points, either $g(e_i) = 0$ or $g(e_i)=\frac12$:
\begin{equation}
\label{eq:max-heat}
h(S_{\mathrm{rf}}) \wwrel\le \max \left\{x + \frac{x^2}{4(x - 2)},\;
x + \frac{1}{2} + \frac{x^2}{4(x - 1)}  \right\}.
\end{equation}
The second case, $g(e_i) > \frac{1}{2}$, is trivial because then there are no possible blocking edges, otherwise one of $G_u$ and $G_v$ would be greater than 1. The edge $e_i$ is then always scheduled immediately upon reaching $x$, and, thus, $h(S_{\mathrm{rf}}) \le x + 1 \le x + \frac{1}{2} + \frac{x^2}{4(x - 1)}$ for $x \ge 2$.
\end{proof}

This bound on~$h(S_{\mathrm{rf}})$ gives us the following corollaries.

\begin{corollary}
\label{col:6-approx}
For Reduce-Fastest$(4)$, we have $h(S_{\mathrm{rf}}) \le 6$.
\end{corollary}
The above corollary gives the upper bound on the maximum heat achieved by Reduce-Fastest(4), which also meets the bound from \wref{thm:6-approx-ratio}.
Optimizing the bound in \weqref{eq:max-heat} allows to slightly improve upon this by changing to $x = 2 + \frac{2\sqrt{5}}{5}$.

\begin{corollary}
\label{col:optimal}
Reduce-Fastest$(2 + \frac{2\sqrt{5}}{5} \approx 2.89)$ gives
$h(S_{\mathrm{rf}}) \le 3 + \sqrt{5} < 5.24$.
\end{corollary}

\subsection{Length of schedules}
\label{sec:length}

In this section, we will show how an arbitrary schedule for an OPS problem can be modified to produce a periodic schedule with a polynomial-sized period and a maximum heat of at most 4 times the maximum heat of the original schedule.

This proof uses an overloading of the definition of heat $h$ such that $h(S, R')$ is the heat achieved by applying some cyclic schedule $S$ to an OPS instance~$\mathcal{P}'_o = (P, R', g)$ where $R'\subseteq R$, and 
$h(S, e, d)$ is the heat of a single edge $e\in R$
that results from applying the same schedule $S$
such that $d$ is the duration since $e$ was last schedled.
\begin{lemma}
\label{lem:length}  
Consider an OPS instance~$\mathcal{P}_o = (P, R, g)$. Given an \textbf{Arbitrary} schedule ~$S_A$ for~$\mathcal{P}_o$, we can construct a \textbf{Polynomial} schedule~$S_P$ with period~$2\chi_{1}(P, R)$ and heat~$h(S_P) \leq 4h(S_A)$.
\end{lemma}

\begin{proof}
Recall that~$r_S(e)$ is the recurrence time of an edge~$e$ in any schedule $S$.
First, consider a \textbf{Truncated} schedule $S_{T}$ which consists of the first $\chi_{1}$ characters of $S_A$.
Let the subset of edges $R_{T}\subseteq R$ consist of the edges $e\in R_{T}$ that appear in $S_{T}$.
Using these, prove the following claim:

\begin{claim}
 \label{cla:int-sch}  
    Given the schedules $S_A$ and $S_{T}$, defined above, it holds that $h(S_{T}, R_{T})\leq 2h(S_A)$.
\end{claim}

\begin{proof}
    For the first  $\chi_1$ days, $S_A$ and $S_{T}$ are identical, so both have heat at most $h(S_A)$ within that period. %
    Any edge $e\in R_{T}$ will have a heat $h(S_{T},e,d)=g(e)\cdot d$, where $d$ is again the duration since $e$ was last scheduled. The set of such durations $D$ includes the duration $d_b$ before the first meeting, the duration $d_a$ after the last meeting, and the durations $D_i = d_1, d_2\ldots d_i$ between meetings, so $D=\{d_b\}\cup \{d_a\}\cup D_i$.
    Note that the heats of edges grow on days when they are scheduled, so the sum of all such durations $\sum_{d\in D}d$ is the period $\chi_{1}$ of $S_T$.

    For the first $\chi_{1}$ days, $S_T$ does not repeat, so it holds that:
    \[
        h(S_A) \wwrel\geq \max_{d\in D}(h(S_{T},e,d)) \wwrel= g(e)\cdot \argmax(\{d_b\}\cup \{d_a\}\cup D_i)    
    \]    

    Thus:
    \begin{equation}
        \label{eq:growth rate of a fragment}
        g(e) \wwrel\leq \frac{h(S_A)}{\argmax(\{d_b\}\cup \{d_a\}\cup D_i)}
    \end{equation}
    After $\chi_{1}$ days, $S_A$ and $S_{T}$ diverge, with $S_{T}$ repeating itself. The maximum heat of $e$ in this period $h(S_{T}, e)$ is therefore given by:
    \[
    h(S_{T}, e) \wwrel= 
    \max_{d\in D}(h(S_{T},e,d)) \wwrel= 
    g(e)\cdot \argmax(\{(d_b+d_a)\}\cup D_i)
    \]
    If $h(S_{T}, e) > h(S_A)$, it follows that $h(S_{T}, e) = g(e)\cdot (d_b+d_a)$. Which, along with \weqref{eq:growth rate of a fragment} shows that for all $e\in R_{T}$:
    \[
        h(S_{T}, e) 
        \wwrel\leq h(S_A) \frac{d_b+d_a}{\argmax(\{d_b\}\cup \{d_a\}\cup D_i)}
        \wwrel\leq 2h(S_A)
    \]
    Hence, $h(S_{T}, R_{T})\leq 2h(S_A)$.
\end{proof}

Let~$C := \chi_1(P, R)$ be the number of colours required for edge colouring the graph~$(P, R)$.
Observe that for an OPS instance~$\mathcal{P}_o$, 
one can always obtain a feasible schedule from a proper edge colouring~$c: e \rightarrow [C]$ of the graph~$(P, R)$, 
where~$e \in R$, by taking any round-robin schedule of the~$C$
colours. 
More formally, we can define a \textbf{coloured} schedule~$S_C(t)$ := $\{e \in R : c(e) \equiv t \, (\text{mod} \, C)\}$. 
As shown by~\cite{GasieniecSmithWild2024}[Proposition 5.1], unweighted OPS is the same problem as edge colouring; thus, the round-robin schedule~$S_C$ of the coloured edges will cover all of the edges.

Finally, we construct the Polynomial schedule~$S_P$ by interleaving the Truncated schedule~$S_{T}$ with the coloured schedule~$S_C$. 
This means that for any edge~$e\in R_T$, the recurrence times obey~$2r_{S_{T}}(e) \leq r_{S_P}(e)$. 
Now consider the complement to $R_T$, $R_T^\complement := R\backslash R_T$. 
As edges $e'\in R_T^\complement$ do not appear in the first $\chi_{1}$ days of $S_A$, their growth rate is bounded by $g(e')<\frac{h(S_A)}{\chi_{1}}$, and they grow to heat $h(S_C, R_T^\complement)<h(S_A)$ over the course of $S_C$.
This completes the claim: for any edge~$e'\in R_T^\complement$, we have~$2r_{S_{C}}(e') = r_{S_P}(e')$.
In other words, the recurrence time of any edge in~$S_P$ is at most twice its recurrence time in either~$S_{T}$ or~$S_{C}$.

This finally demonstrates that~$h(S_P) \leq 2h(S_{C})\leq 2h(S_{T}, R_T)$ and therefore, from \wref{cla:int-sch}, we have~$h(S_P) \leq 4h(S_A)$.
\end{proof}

\begin{corollary}
    As BGT is a special case of OPS on star graphs, \wref{lem:length} also applies to BGT.
\end{corollary}

\begin{remark}
    Only the first $\chi_{1}$ days of $S_A$ appear in $S_P$, and only edges $e'\in R_T^\complement$ need to appear in $S_C$. Further works may apply this process to unsustainable schedules, taking advantage of these facts to reduce overall heat.
\end{remark}

%
%
%
%
%
%
%
%
%
%
%
%
%
%
%
%
%
%
%
%
%
%
%
%
%
%
%
%
%
%
%
%
%
%
%
%
%
%
%
%
%
%
%
%


\section{Poly Density}
\label{sec:density}
This section introduces two key results: \wref{sec:ops-density} provides improved upper and lower bounds by proving \wref{thm:OPS_density_combined}, while \wref{sec:dps-density} shows that poly density for DPS generalizes density for PWS and proves \wref{thm:densityThreshold}, showing that OPS instances with $d\leq \frac14$ are schedulable.

\subsection{Bounding Poly Density for OPS polycules}
\label{sec:ops-density}
Recall the notion of poly density introduced by~\cite{GasieniecSmithWild2024}, which follows the same intuition as Pinwheel Scheduling: relax the requirement that one complete task must be done each day, instead allowing any fractions of a day to be spent on each task. Use this relaxation to identify classes of schedulable or unschedulable instances. Pinwheel Scheduling uses an implicit schedule of unit length in which each task $i$ with frequency $f_i$ has contribution $y_i \geq \frac{1}{f_i}$. Task $i$ is then considered to be performed after $\frac{1}{y_i}$ days, with the density of the PWS instance given by the sum of these contributions $d=\sum_{i=1}^k \frac1{f_i}$. If $d>1$ then these contributions cannot fit into a single day, and the instance clearly cannot be scheduled. 

\cite{GasieniecSmithWild2024} applied similar reasoning to determine the poly density of Optimisation Polyamorous Scheduling problems, obtaining a lower bound on the optimal heat $h^*$. 
First, compose an optimisation problem using the same logic as Pinwheel Scheduling: construct a fractional schedule of unit length by assigning contribution $y_M\in[0,1]$ to each matching $M$. The heat $\bar h$ of the best such schedule provides a lower bound on $h^*$.

\begin{alignat}{4}
& \text{min}\ \bar{h} && && &&\\
\label{eq:minhbar}
& \text{s.t.}    & \sum_{M \in \mathcal{M}} y_M     &\;\leq\; 1&&\\
&& \frac{g(e)}{\sum_{\substack{M\in \mathcal{M}: e \in M}}y_M} & \;\leq\;\bar{h}&&\qquad\qquad\forall e \in R\\
& &y_M &\;\in\; [0,1]
\label{eq:ymin01}
\end{alignat}

\cite{GasieniecSmithWild2024} defines the poly density $\bar{h}^*$ of an OPS polycule $\mathcal{P}_o = (P, R, g)$ as the optimal objective value of the fractional-schedule LP.%
\footnote{A change of variables $\ell = 1/\bar h$ and taking reciprocals in the second constraint yields this LP.}
From the dual of this LP, \cite{GasieniecSmithWild2024} obtain a more explicit term for $\bar h^*$. A dual solution consists of weights $z_e$ for edges $e\in R$ with $\sum_{e\in R} z_e =1$. Intuitively these account for their growth rates $g_e$ and the (inclusion-maximal) matchings $M \in \mathcal{M}$ in which they appear. 
Any such weighting implies a lower bound $\bar{h}(z)$ on $\bar h^*$ and hence on $h^*$, 
with the best choice of $z$ yielding the poly density
\begin{equation}
\bar{h}^* 
\wrel= \max_{z\in \mathcal Z} \bar{h}(z)
\text{, \qquad where }
\bar{h}(z) \wrel= \min_{M \in \mathcal{M}} {\dfrac{1}{ \sum_{e \in M} \frac{z_e}{g(e)}}},
\end{equation}
where $\mathcal Z$ is the set of all weightings.

Recall from \wref{sec:preliminaries} that the maximum personal growth over some polycule $\mathcal{P}_o = (P, R, g)$ is given by $\Gstar = \max_{v \in P} \sum_{e \in R:v \in e}g(e)$. 
We now show that the poly density is within a constant factor of $\Gstar$,
proving \wref{thm:OPS_density_combined} in two parts: The lower bound by \wref{lem:Gstar-le-hbarstar} and the upper bound by \wref{lem:density-bound-general}.

\begin{lemma}
\label{lem:Gstar-le-hbarstar}
    For any OPS instance $\mathcal{P}_o = (P, R, g)$, the poly density $\bar{h}^*$ is bounded by $\bar{h}^* \geq\Gstar$. 
\end{lemma}
\begin{proof}
    Let $v^*$ be the vertex with the maximum local heat $\Gstar$, \ie{} 
    \[v^* \wrel= \arg\max_{v\in P} \sum_{e \in R:v\in e} g(e).\] 
    Consider the weighting $z^*$ with $z_e=\frac{g(e)}{\Gstar}$ when $v^*\in e$ and $z_e=0$ otherwise.
    For $z^*$, it holds that 
    \[\max_{M \in \mathcal{M}} \sum_{e \in M} \frac{z_e}{g(e)} \wrel= \frac{1}{\Gstar}\]
    and $\bar{h}(z)=\Gstar$.
    Since $\bar{h}^* \geq \bar{h}(z)$ for all weights $z$, the inequality holds.
\end{proof}

\begin{lemma}
    \label{lem:density-bound}
    For any OPS instance $\mathcal{P}_o = (P, R, g)$ with integer growth rates, $\bar{h}^* \leq \frac{3}{2}\Gstar$.
\end{lemma}
\begin{proof}
    Begin by constructing an unweighted multigraph $\mathcal{P}'_o=(P,E)$ such that each edge $e \in R$ is replaced by a bundle of $g(e)$ parallel edges in $E$, noting that the maximum degree $\Delta$ of $\mathcal{P}'_o$ is exactly $\Gstar$. By \cite{shannon1949theorem}, we can colour edges in $\mathcal{P}'_o$ using $C\leq\frac{3}{2}\Delta$ colours.
    Given that each colour induces a matching in $\mathcal{P}_o$ such that each edge $e \in R$ appears in exactly $g(e)$ colours,
    construct a fractional solution by assigning a contribution $y_M = \frac{1}{C}$ to each colour class $M$.
    Checking Equations \eqref{eq:minhbar}\,--\,\eqref{eq:ymin01} shows that this is a valid fractional schedule with heat 
    \[\bar{h} \wwrel\geq {g(e)}\bigg/{\sum_{\substack{M\in \mathcal{M}: e \in M}}y_M} \wwrel= C\] 
    for all edges. 
    Hence, $\bar{h}=C\leq \frac{3}{2}\Gstar$ for this assignment, and $\bar{h}^* \leq \frac{3}{2}\Gstar$ globally.%
\end{proof}

\begin{lemma}
    \label{lem:density-bound-general}
    The poly density $\bar{h}^*$ of any OPS instance $\mathcal{P}_o = (P, R, g)$  is bounded by $\bar{h}^* \leq \frac{3}{2}\Gstar+\epsilon$ for any $\epsilon>0$.
\end{lemma}
\begin{proof}
    When all heat growth rates are rational numbers, we can scale all of them by their common denominator to
    reduce them to the integral case. 
    If any growth rate is irrational, we can upper-bound it with a rational number that is arbitrarily close.
\end{proof}

\subsection{Poly Density of DPS polycules}
\label{sec:dps-density}
In this section, we turn our attention to to the poly density of Decision Polyamorous Scheduling: showing that it generalizes the density of Pinwheel Scheduling with \wref{thm:DPS-star-density}, introducing the PolyGreedy algorithm and analysing it to prove \wref{thm:poly1/2} and that $d\leq \frac{1}{4}$ implies schedulability with a proof for \wref{thm:densityThreshold}.
When introduced by \cite{GasieniecSmithWild2024}, poly density was chiefly discussed in the context of OPS, but the poly density of a DPS instance $\mathcal{P}_d = (P, R, f)$ was defined as the poly density of a paired OPS instance $\mathcal{P}_o = (P, R, \frac{1}{f})$: 

\begin{definition}[Poly Density of DPS instances] 
\label{def:DPS-density}
The poly density $d$ of a DPS polycule $\mathcal{P}_d = (P, R, f)$ is given by
\[d \wwrel= \max_{z\in \mathcal Z} \dfrac{1}{\displaystyle\max_{M \in \mathcal{M}} \sum_{e \in M} z_e\cdot f_e}\]
where $\mathcal Z$ is the set of all weightings applying $z_e \in [0,1]$ to each edge $e\in R$ such that $\sum_{e\in R} z_e = 1$ and $\mathcal{M}$ is the set of inclusion maximal matchings $M$ in the graph $(P,R)$.
\end{definition}

Given the equivalence between DPS polycules on star graphs $\mathcal{P}^*_d = (P, R, f)$  and Pinwheel Scheduling instances with the same frequency set $f$, this is a generalisation of the density of Pinwheel Scheduling:

\begin{theorem}
    \label{thm:DPS-star-density}
    The poly density $d$ of any DPS polycule $\mathcal{P}^*_d = (P, R, f)$ whose graph $(P,R)$ is a star is given by $d = \sum_{e\in R} \frac{1}{f_e}$
\end{theorem}
\begin{proof}
    Begin with the poly density defined by \wref{def:DPS-density}. As $\mathcal{P}^*_d$ is a star graph, each matching $m\in \mathcal{M}$ contains a single edge $e\in R$, so $d=\max_{z\in \mathcal Z} 1\big/ (\max_{e \in R} z_e\cdot f_e)$. This will be solved by the value of $z$ which minimizes $\max_{e \in R} z_e\cdot f_e$. As $\sum_{e\in R}z_e = 1$, any assignment $\mathcal Z$ which does not set $z_e\cdot f_e = z_{e'}\cdot f_{e'}$ for all $e,e'\in R$ will have a higher maximum than one that does. Try 
    \[z_e \wrel= \dfrac{1}{f_e\sum_{e'\in R}\frac{1}{f_{e'}}},\] 
    observing that 
    \[
        \sum_{e\in R}z_e 
        \wwrel= 
        \dfrac{1}{\displaystyle \sum_{e'\in R}\frac{1}{f_{e'}}}\cdot \sum_{e\in R}\frac{1}{f_e} 
        \wwrel= 
        1,
    \] 
    and that $d=\sum_{e\in R} \frac{1}{f_e}$.
\end{proof}

We continue by identifying two classes of schedulable DPS instances which instead restrict the \defi{local density} of vertices~-- the density of the portion of the instance directly connected to any single vertex. This analysis is motivated by~\cite{holte1989pinwheel} and uses the following PolyGreedy algorithm for DPS:

\begin{algorithm}[h]
	\caption{$\mathsf{PolyGreedy}(P,R=\{e_1,...,e_{|R|}\}, f)$}
	\label{alg:PolyGreedy}
	\begin{algorithmic}[1]
		\State /* Assumes $f(e_i) = 2^{k_i}$ for some $k_i\in\N$ for all $e_i \in R$ */
		\State /* W.l.o.g. let $f(e_1) \leq f(e_2) \leq \cdots \leq f(e_{|R|})$ */
		\State $f_{\max} = f(e_{|R|})$
		\For {$t = 0,\ldots,f_{\max}$}
		 \State $S[t] \gets \emptyset$
		\EndFor
		\For {$i=1,\ldots,|R|$}
			\State $s_i \gets \min \bigl\{ t \in [0..f_{\max}) : S[t]\cup\{e_i\} \text{ is a matching} \bigr\}$
			\For {$k=1,\ldots,{f_{\max}}/{f_i}$}
			\State $t \gets s_i+(k-1) \cdot f_i$
			\State $S[t] \gets S[t] \cup \{ e_i \}$
			\EndFor
		\EndFor
		\State \Return $S[0..f_{\max})$
	\end{algorithmic} 
\end{algorithm}

\begin{theorem} \label{thm:poly1/2}
    Let $\mathcal{P}_d = (P, R, f)$ be a DPS instance such that 
    \begin{enumerate}
        \item $\forall e \in R$, $f_e$ is a power of two
        \item $\forall v \in V$, $\displaystyle \sum_{e\in R:v \in e} \frac{1}{f_e} \leq \frac{1}{2}$.
    \end{enumerate}
    Then, there exists a valid schedule for $\mathcal{P}_d$.
\end{theorem}

\begin{proof}
     Let $f_{\max}$ be a maximum frequency among all edges and construct a periodic schedule of size $f_{\max}$ using \wref{alg:PolyGreedy}. 
    By \wref{def:DPS}, valid schedules lack conflicts and respect frequencies; we proceed by showing inductively that \wref{alg:PolyGreedy} respects both properties.
    
    In the initial state, $S$ is empty and thus contains no conflicts. Assuming that there are no conflicts after $i-1$ edges have been added, observe that adding edge $e_i$ initially places it in a conflict-free slot $s_i$, then in slots $s_i+f_i, s_i+2f_i, s_i+3f_i,\ldots$. As all frequencies $f_e$ are powers of two and all edges placed thus far obey $f_j\leq f_i$, slots $s_i+f_i, s_i+2f_i, s_i+3f_i,\ldots$ must also be conflict-free for $e_i$.

    If we again consider an arbitrary edge $e_i$ at any point in $S$ after its initial inclusion $s_i$, it is clear that the schedule is frequency respecting. Initially, $e_i$ is placed in the first slot with no edges that conflict with $e_i$. If $e_i = \{u,v\}$, the number of slots with conflicting edges is at most
    \begin{align*}
        \sum_{\substack{e_j:u \in e_j \\ f_j < f_i}} \frac{f_i}{f_j}+\sum_{\substack{e_j:v \in e_j \\ f_j < f_i}} \frac{f_i}{f_j} \quad&=\quad f_i \Biggl(\sum_{\substack{e_j:u \in e_j \\ f_j < f_i}} \frac{1}{f_j}+\sum_{\substack{e_j:v \in e_j \\ f_j < f_i}} \frac{1}{f_j}\Biggr)\\ &<\quad f_i \biggl(\frac{1}{2}+\frac{1}{2}\biggr)\\ &=\quad f_i
    \end{align*}
    
    Hence, there must be at least one slot $s_i$ among the first $f_i$ slots with no conflicting edges.  
    The last occurrence of $e_i$ in $S$ is at $s_i + f_{\max} - f_i$, so the recurrence time $r$ between the last and first instances of $e_i$ in the periodic solution is given by $f_i$
\end{proof}

\begin{proofof}{\wref{thm:densityThreshold}}
    Begin by constructing an OPS instance $\mathcal{P}_o' = (P, R, g)$ with $g_e = \frac{1}{f_e}$ for all edges $e\in R$; by \wref{def:DPS-density},
    the poly density $\bar{h}*$ of $\mathcal{P}_o'$ equals the density $d$ of the corresponding $\mathcal{P}_d = (P, R, f)$. 
    By \wref{lem:Gstar-le-hbarstar}, the maximum personal growth rate $\Gstar$ of $\mathcal{P}_o'$ is at most $\frac{1}{4}$, so 
    \[
        \frac{1}{4} 
        \wrel\geq 
        \Gstar 
        \wrel= 
        \max_{v \in P}\sum_{e \in R:v \in e} g(e) 
        \wrel= 
        \max_{v \in P} \sum_{e \in R:v \in e}\frac{1}{f_e},
    \] 
    and $\forall v \in V$, $\sum_{e\in R:v \in e} \frac{1}{f_e} \leq \frac{1}{4}$.

    Now construct $\mathcal{P}'_d = (P, R, f')$ such that $\forall e\in R, f'(e) = 2^{\lfloor\log_2 f_e\rfloor}$, \ie{} round all the frequencies down to the nearest power of two such that $\sum_{e:v \in e} \frac{1}{f'_e} \leq \sum_{e:v \in e} \frac{2}{f_e} \leq 2 \cdot\frac{1}{4} = \frac{1}{2}$. Applying \wref{thm:poly1/2} then gives a feasible schedule for $\mathcal{P}'_d$, which must also solve $\mathcal{P}_d$.
\end{proofof}

%
%
%
%
%
%
%
%
%
%
%
%
%
%
%
%
%
%
%
%
%
%
%
%
%
%
%
%
%
%
%
%
%
%
%
%
%
%
%
%
%
%
%
%
%
%
%
%
%
%
%
%
%
%
%
%
%
%
%
%
%
%
%
%
%
%
%
%
%
%
%
%
%
%
%
%
%
%
%
%
%
%
%
%
%
%
%
%
%
%


\section{Inapproximability}
\label{sec:inapproximability}
In this section, we prove \wref{thm:sat-inapprox}; showing, via a reduction from 3SAT, that bipartite OPS does not allow efficient $(1+\delta)$-approximation algorithms for $\delta<\frac{1}{12}$ unless P $=$ NP.

\subsection{Overview of Proof}
\label{sec:inapproximability-overview}

The proof of \wref{thm:sat-inapprox} has two steps: a reduction from 3SAT to bipartite DPS, and a conversion to bipartite OPS such that a $\frac{13}{12}$ approximation is shown to be NP-hard.

3SAT is the following problem: given a 3-CNF formula $\varphi = c_1\land c_2\land \cdots \land c_m$ with clauses $C=\{c_1,c_2,\ldots, c_m\}$, each consisting of exactly 3 literals over variables $X=\{x_1, x_2,\ldots, x_{n}\}$, decide whether there is an assignment of Boolean values to the variables in $X$ such that all $m$ clauses in $C$ are satisfied. 

\begin{lemma}[3SAT $\le_p$ bipartite DPS]
    \label{lem:dps=3sat}
    For any 3-CNF formula $\varphi$ with $m$ clauses, we can construct in polynomial time a bipartite decision polycule $\mathcal P_{d\varphi}$ which has a valid schedule if and only if all clauses of $\varphi$ can be simultaneously satisfied.
\end{lemma}

The bulk of \wref{sec:inapproximability} will introduce our construction of $\mathcal P_{d\varphi}$, which we will use to prove this Lemma. This construction is designed to be modular and communicable, rather than efficient with constant factors.

The second step in the proof of \wref{thm:sat-inapprox} is to convert the decision polycule $\mathcal P_{d\varphi}$ from \wref{lem:dps=3sat} to an optimisation polycule $\mathcal P_{o\varphi}$ using \wref{lem:dps-ops}.
It will be immediately obvious from the construction that the largest frequency in $\mathcal P_{d\varphi}$ is $F=12$.
So by \wref{lem:dps-ops}, $\mathcal P_{o\varphi}$ either has $h^* = 1$ (if $\mathcal P_{d\varphi}$ is feasible) or else has $h^*\ge \frac{13}{12}$.

\begin{proofof}{\wref{thm:sat-inapprox}}
Assume that there is a polynomial-time approximation algorithm $A$, which solves bipartite OPS with approximation ratio $\alpha<\frac{13}{12}$. 
If $\mathcal P_{o\varphi}$ has optimal heat $h^* = 1$, $A$ produces a schedule with heat $1\leq h\leq \alpha < \frac{13}{12}$; 
whereas if $h^*\ge \frac{13}{12}$, $A$ must produce a schedule of heat $\frac{13}{12}\leq h\leq \alpha\cdot\frac{13}{12}$.
By running $A$, we are therefore able to distinguish between $h^*\le 1$ and $h^*\ge \frac{13}{12}$ for $\mathcal P_{o\varphi}$, and hence between the feasibility or infeasibility of $\mathcal P_{d\varphi}$. 
This then discriminates between Yes and No instances of 3SAT, via the polynomial-time reduction from \wref{lem:dps=3sat}.
3SAT is NP-hard, hence P $=$ NP follows.
\end{proofof}

Let us denote by $\alpha^*$ the approximability threshold for bipartite OPS,
that is:
efficient polynomial-time approximation algorithms with approximation ratio $\alpha$ exist 
if and only if $\alpha \ge \alpha^*$ (assuming P $\ne$ NP).
\wref{thm:sat-inapprox} shows that $\alpha^*\geq\frac{13}{12}$, 
while \wref{thm:5.24-approx-ratio} shows that $\alpha^* \le 5.24$,
leaving a substantial gap. 
We conjecture that the constant $\frac{13}{12}$ can be improved by careful analysis of our construction, but we leave this to future work. 

\begin{conjecture}
    \label{conj:4/3}
    The approximability threshold for bipartite OPS is $\alpha^* \ge \frac43$.
\end{conjecture}

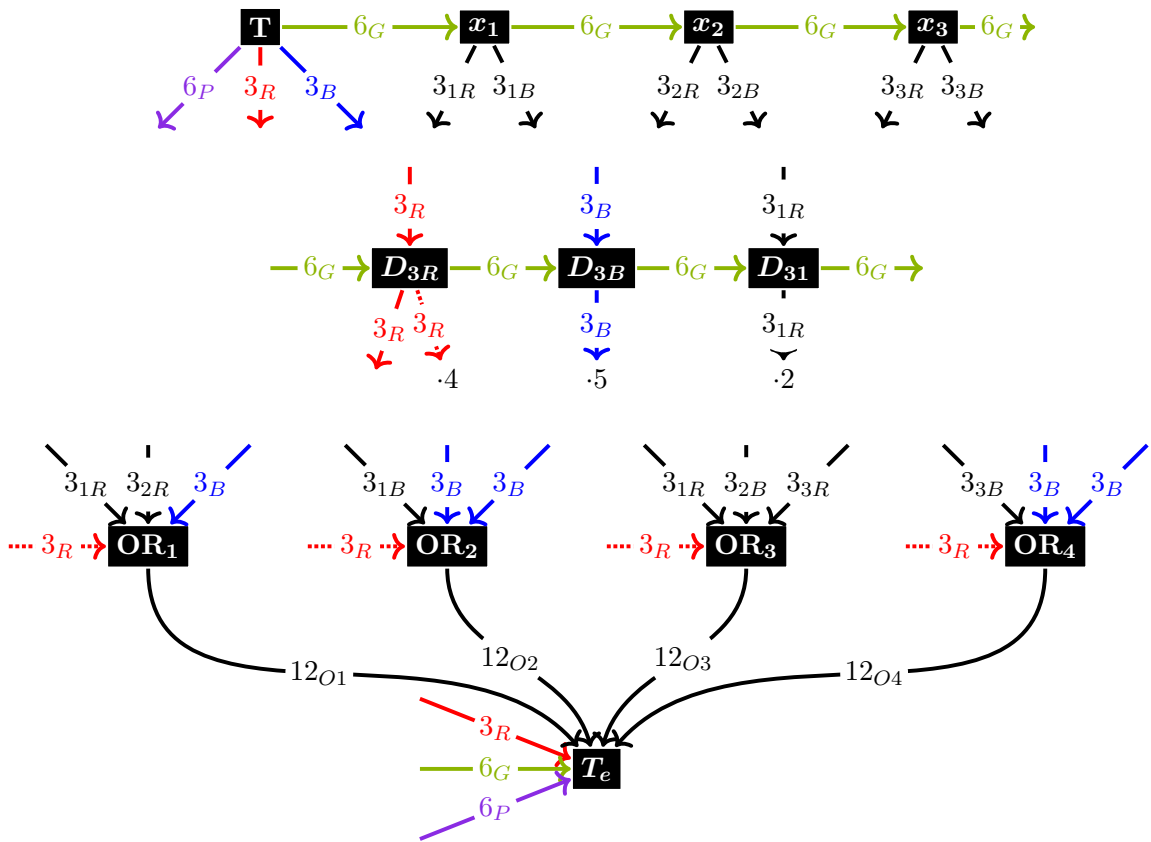
\begin{figure}[htbp] %
    \centering
    \resizebox{\textwidth}{!}{
        \begin{tikzpicture}[scale=2]

        \def\xadj{0.75}
        \node[shorthand] (T) at (-1+\xadj, 0) {T};
        \node[shorthand] (x1) at (0.5+\xadj, 0) {$x_1$};
        \node[shorthand] (x2) at (2+\xadj, 0) {$x_2$};
        \node[shorthand] (x3) at (3.5+\xadj, 0) {$x_3$};
        \node[invisible node] (spare out) at (4.25+\xadj, 0) {};
    
        \path[->] (T) edge[std, ourGreen] node[edge descriptor]  {\color{applegreen}$6_G$} (x1);
        \path[->] (x1) edge[std, ourGreen] node[edge descriptor] {\color{applegreen}$6_G$} (x2);
        \path[->] (x2) edge[std, ourGreen] node[edge descriptor] {\color{applegreen}$6_G$} (x3);
        \path[->] (x3) edge[std, ourGreen] node[edge descriptor] {\color{applegreen}$6_G$} (spare out);
    
        \node[invisible node] (T out 1) at (-1.75+\xadj, -0.75) {};    
        \node[invisible node] (T out 2) at (-1+\xadj, -0.75) {};
        \node[invisible node] (T out 3) at (-0.25+\xadj, -0.75) {};
        \node[invisible node] (x1 out 1) at (0.125+\xadj, -0.75) {};
        \node[invisible node] (x1 out 2) at (0.875+\xadj, -0.75) {};
        \node[invisible node] (x2 out 1) at (1.625+\xadj, -0.75) {};
        \node[invisible node] (x2 out 2) at (2.375+\xadj, -0.75) {};
        \node[invisible node] (x3 out 1) at (3.125+\xadj, -0.75) {};
        \node[invisible node] (x3 out 2) at (3.875+\xadj, -0.75) {};
    
        \path[->] (T) edge[std, ourPurple] node[edge descriptor] {\color{blue-violet}$6_P$} (T out 1);
        \path[->] (T) edge[std, ourRed] node[edge descriptor] {\color{red}$3_R$} (T out 2);
        \path[->] (T) edge[std, ourBlue] node[edge descriptor] {\color{blue}$3_B$} (T out 3);
        \path[->] (x1) edge[std, ourBlack] node[edge descriptor] {$3_{1R}$} (x1 out 1);
        \path[->] (x1) edge[std, ourBlack] node[edge descriptor] {$3_{1B}$} (x1 out 2);
        \path[->] (x2) edge[std, ourBlack] node[edge descriptor] {$3_{2R}$} (x2 out 1);
        \path[->] (x2) edge[std, ourBlack] node[edge descriptor] {$3_{2B}$} (x2 out 2);
        \path[->] (x3) edge[std, ourBlack] node[edge descriptor] {$3_{3R}$} (x3 out 1);
        \path[->] (x3) edge[std, ourBlack] node[edge descriptor] {$3_{3B}$} (x3 out 2);

        \def\xadj{0}
        \def\yadj{-0.125}
       
        \node[invisible node] (2inR)    at (0.75+\xadj,             -0.75+\yadj) {};    
        \node[invisible node] (2inB)    at (2+\xadj,                -0.75+\yadj) {};
        \node[invisible node] (2in1)    at (3.25+\xadj,             -0.75+\yadj) {};

        \node[shorthand] (D3R)          at (0.75+\xadj,             -1.5+\yadj) {$D_{3R}$};
        \node[shorthand] (D3B)          at (2+\xadj,                -1.5+\yadj) {$D_{3B}$};
        \node[shorthand] (D31)          at (3.25+\xadj,             -1.5+\yadj) {$D_{31}$};

        \path[->] (2inR) edge[std, ourRed] node[edge descriptor] {\color{red}$3_R$} (D3R);
        \path[->] (2inB) edge[std, ourBlue] node[edge descriptor] {\color{blue}$3_B$} (D3B);
        \path[->] (2in1) edge[std, ourBlack] node[edge descriptor] {$3_{1R}$} (D31);

        \node[invisible node] (2inG)    at (-0.25+\xadj,            -1.5+\yadj) {};
        \node[invisible node] (2outG)   at (4.25+\xadj,             -1.5+\yadj) {};

        \path[->] (2inG) edge[std, ourGreen] node[edge descriptor] {\color{applegreen}$6_G$} (D3R);
        \path[->] (D31) edge[std, ourGreen] node[edge descriptor] {\color{applegreen}$6_G$} (2outG);
        \path[->] (D3R) edge[std, ourGreen] node[edge descriptor] {\color{applegreen}$6_G$} (D3B);
        \path[->] (D3B) edge[std, ourGreen] node[edge descriptor] {\color{applegreen}$6_G$} (D31);

        \node[invisible node] (2outR L)     at (0.75+\xadj-0.25,    -2.75+0.5+\yadj) {};    
        \node[invisible node] (2outB L)     at (2+\xadj,            -2.75+0.5+\yadj) {$\cdot 5$};
        \node[invisible node] (2out1)       at (3.25+\xadj,         -2.75+0.5+\yadj) {$\cdot 2$};
        \node[invisible node] (2outR R)     at (0.75+\xadj+0.25,    -2.75+0.5+\yadj) {$\cdot 4$};

        \path[->] (D3R) edge[std, ourRed] node[edge descriptor] {\color{red}$3_R$} (2outR L);
        \path[->] (D3B) edge[std, ourBlue] node[edge descriptor] {\color{blue}$3_B$} (2outB L);
        \path[->] (D31) edge[std, ourBlack] node[edge descriptor] {$3_{1R}$} (2out1);
        \path[->] (D3R) edge[std B, ourRed] node[edge descriptor] {\color{red}$3_R$} (2outR R);

        \def\yadj{1.5}        
        \node[invisible node] (OR1 in top 1)    at (-1.75, -4.25+\yadj) {};
        \node[invisible node] (OR1 in top 2)    at (-1, -5.25+1+\yadj) {};
        \node[invisible node] (OR1 in top 3)    at (-0.25, -5.25+1+\yadj) {};
        \node[shorthand] (OR1)                  at (-1, -6+1+\yadj) {OR$_1$};
        \node[invisible node] (OR1 in red)      at (-2, -6+1+\yadj){};
        
        \path[->] (OR1 in top 1) edge[std, ourBlack] node[edge descriptor] {$3_{1R}$} (OR1);
        \path[->] (OR1 in top 2) edge[std, ourBlack] node[edge descriptor] {$3_{2R}$} (OR1);
        \path[->] (OR1 in top 3) edge[std, ourBlue] node[edge descriptor] {$\color{blue}3_B$} (OR1);
        \path[->] (OR1 in red)   edge[std B, ourRed] node[edge descriptor] {\color{red}$3_R$} (OR1);

        \node[invisible node] (OR2 in top 1)    at (0.25, -5.25+1+\yadj) {};
        \node[invisible node] (OR2 in top 2)    at (1, -5.25+1+\yadj) {};
        \node[invisible node] (OR2 in top 3)    at (1.75, -5.25+1+\yadj) {};
        \node[shorthand] (OR2)                  at (1, -6+1+\yadj) {OR$_2$};
        \node[invisible node] (OR2 in red)      at (0, -6+1+\yadj){};
        
        \path[->] (OR2 in top 1) edge[std, ourBlack] node[edge descriptor] {$3_{1B}$} (OR2);
        \path[->] (OR2 in top 2) edge[std, ourBlue] node[edge descriptor] {$\color{blue}3_B$} (OR2);
        \path[->] (OR2 in top 3) edge[std, ourBlue] node[edge descriptor] {$\color{blue}3_B$} (OR2);
        \path[->] (OR2 in red)   edge[std B, ourRed] node[edge descriptor] {\color{red}$3_R$} (OR2);

        \node[invisible node] (OR3 in top 1)    at (2.25, -5.25+1+\yadj) {};
        \node[invisible node] (OR3 in top 2)    at (3, -5.25+1+\yadj) {};
        \node[invisible node] (OR3 in top 3)    at (3.75, -5.25+1+\yadj) {};
        \node[shorthand] (OR3)                  at (3, -6+1+\yadj) {OR$_3$};
        \node[invisible node] (OR3 in red)      at (2, -6+1+\yadj){};
        
        \path[->] (OR3 in top 1) edge[std, ourBlack] node[edge descriptor] {$3_{1R}$} (OR3);
        \path[->] (OR3 in top 2) edge[std, ourBlack] node[edge descriptor] {$3_{2B}$} (OR3);
        \path[->] (OR3 in top 3) edge[std, ourBlack] node[edge descriptor] {$3_{3R}$} (OR3);
        \path[->] (OR3 in red)   edge[std B, ourRed] node[edge descriptor] {\color{red}$3_R$} (OR3);

        \node[invisible node] (OR4 in top 1)    at (4.25, -4.25+\yadj) {};
        \node[invisible node] (OR4 in top 2)    at (5, -5.25+1+\yadj) {};
        \node[invisible node] (OR4 in top 3)    at (5.75, -5.25+1+\yadj) {};
        \node[shorthand] (OR4)                  at (5, -6+1+\yadj) {OR$_4$};
        \node[invisible node] (OR4 in red)      at (4, -6+1+\yadj){};
        
        \path[->] (OR4 in top 1) edge[std, ourBlack] node[edge descriptor] {$3_{3B}$} (OR4);
        \path[->] (OR4 in top 2) edge[std, ourBlue] node[edge descriptor] {$\color{blue}3_B$} (OR4);
        \path[->] (OR4 in top 3) edge[std, ourBlue] node[edge descriptor] {$\color{blue}3_B$} (OR4);
        \path[->] (OR4 in red)   edge[std B, ourRed] node[edge descriptor] {\color{red}$3_R$} (OR4);

        \def\yadjnew{7.5}

        \node[shorthand]      (tension)     at (2, -12.5+\yadjnew){$T_e$};
        \node[invisible node] (red)         at (0.75, -0.25-11.75+\yadjnew) {};
        \node[invisible node] (green)       at (0.75, -0.75-11.75+\yadjnew) {};
        \node[invisible node] (purple)      at (0.75, -1.25-11.75+\yadjnew) {};

        \path[->] (red) edge[std, ourRed] node[edge descriptor] {$\color{red}3_R$} (tension);
        \path[->] (green) edge[std, ourGreen] node[edge descriptor] {$\color{applegreen}6_G$} (tension);
        \path[->] (purple) edge[std, ourPurple] node[edge descriptor] {$\color{blue-violet}6_P$} (tension);
        \path[<-] (tension) edge[in = -90, out = 130, std, ourBlack] node[edge descriptor] {$12_{O1}$} (OR1);
        \path[<-] (tension) edge[in = -90, out = 105, std, ourBlack] node[edge descriptor] {$12_{O2}$} (OR2);
        \path[<-] (tension) edge[in = -90, out = 50, std, ourBlack] node[edge descriptor] {$12_{O4}$} (OR4);
        \path[<-] (tension) edge[in = -90, out = 75, std, ourBlack] node[edge descriptor] {$12_{O3}$} (OR3);
    \end{tikzpicture}
    }
    
    \caption{
            A DPS polycule which is schedulable iff there is some assignment that satisfies $(x_1\lor x_2), (\overline{x_1}), (x_1 \lor \overline{x_2} \lor x_3),$ and $ (\overline{x_3})$ (which is not possible).
            Connections between layers are omitted for clarity but flow from top to bottom, starting with the variable layer, then a layer duplicating nodes with frequency 3, the \OR layer, and finally the tensioning layer. Note that larger problems will also require a $D_6$ layer to support additional tension gadgets.
    }

    \label{fig:worked-eg:Pdsz}
\end{figure} %

\subsection{Overview of Reduction}
\label{sec:DMAX3SAT-to-OPS}

We now give the proof of \wref{lem:dps=3sat}.
For the remainder of this section, we will assume that a 3-CNF formula $\varphi = c_1\land \cdots\land c_m$ over variables $X = \{x_1,\ldots, x_{n}\}$ is given.
Using this, we will show how to construct DPS instances $\mathcal P_{d\varphi}$ that admits a schedule iff
there is a variable assignment $v:X \to \{\mathrm{True},\mathrm{False}\}$ that satisfies all clauses in $C = \{c_1,\ldots,c_m\}$. 
This construction directly represents components of Boolean formulas as bipartite DPS ``gadgets'': 
\begin{itemize}
\item \emph{variables} (\wref{sec:variables}),
\item clauses (\emph{OR gadgets}) (\wref{sec:OR}), and 
\item a check that all clauses are True, the \emph{tension gadget} (\wref{sec:Tension}).
\end{itemize}
To make these gadgets work, we require further auxiliary gadgets:
\begin{itemize}
\item a \emph{``True Clock''} to break ties between symmetric choices for schedules (\wref{sec:TrueClock}),
\item slot duplication gadgets: \emph{$D_3$ duplicators}(\wref{sec:D3}), and \emph{$D_6$ duplicators} (\wref{sec:D6}).
\end{itemize}

The overall conversion algorithm is given by \wref{def:gadgets-assemble} and
a worked example is shown in \wref{fig:worked-eg:Pdsz}.

\begin{definition}[$\mathcal P_{d\varphi}$ polycules]
	\label{def:gadgets-assemble}
	The decision polycule $\mathcal P_{d\varphi}=(P,R,f)$ is constructed in layers as follows:
	\begin{thmenumerate}{def:gadgets-assemble}
	\item \textbf{Variable layer:}\\
		The variable layer consists of a variable gadget with outputs $3_{iR}$ and $3_{iB}$ for each variable $x_i\in X$, as well as a single special variable: the True Clock.
	\item \textbf{Duplication layer:}\\
		The duplication layer duplicates outputs of the variable layer: 
		$D_3$ duplicators create one $3_{iR}$ edge for each $x_i$ in clauses $c\in C$, and one $3_{iB}$ edge for each occurrence of $\overline{x_i}$ in clauses $c\in C$.
		They also create as many $\color{red}3_R$ and $\color{blue}3_B$ edges as are needed, while $D_6$ duplicators do the same for $\color{applegreen}6_G$ and $\color{blue-violet}6_P$ edges.
		Both male and female constant edges are created, and unused edges are connected to pendent nodes.
	\item \textbf{Clause layer:}\\
		The clause layer consists of one \OR gadget for each clause $c_j\in C$. \OR gadgets have three inputs, each corresponding to a literal in $c_j$: $3_{iR}$ for $x_i$, $3_{iB}$ for $\overline{x_i}$. 
		For clauses with less than 3 literals, $\color{blue}3_B$ edges fill the \OR gadget's unused inputs.
	\item \textbf{Tension layer:}\\
		The tensioning layer attaches tension gadgets to all outputs of the sorting layer; any spare inputs to the tensioning layer are connected to pendant nodes.
	\end{thmenumerate}
\end{definition}

\subsection{The True Clock \& Colour Slots}
\label{sec:TrueClock}

\wref{fig:variables} introduces gadgets representing sample variables $x_1$ and $x_2$, as well as a special variable: the True Clock $T$, which acts as a drumbeat for the polycule as a whole. 
Variables, including $T$, have four relationships: $[3, 3, 6, 6]$, so their local schedules must all have the following form: $[\splitaftercomma{3_a, 3_b, 6_a, 3_a, 3_b, 6_b}]$. 
As schedules are cyclic, we can choose to start this schedule with the $3_a$ edge of $T$ without loss of generality. We then assign names to these edges as dictated by the slots they are in:

\begin{definition}[Slots]
    \label{def:slots}
	We call days $t\in\mathbb{N}_0$ with 
 $t\equiv 0 
 \pmod 3$ 
 the \textcolor{red}{red} slots, days $t\equiv 1\pmod 3$  
 \textcolor{blue}{blue} slots, days $t\equiv 2\pmod 6$ 
 \textcolor{applegreen}{green} slots, and days $t\equiv 5 \pmod 6$ 
 \textcolor{blue-violet}{purple} slots.
\end{definition}

The local schedule of the True Clock is therefore given by
[\color{red}$3_R$\color{black},
\color{blue}$3_B$\color{black},
\color{applegreen}$6_G$\color{black},
\color{red}$3_R$\color{black},
\color{blue}$3_B$\color{black},
\color{blue-violet}$6_P$\color{black}].
As $\color{red}3_R$ is scheduled on day 0, it will always be assigned in red slots, 
with $\color{blue}3_B$, $\color{applegreen}6_G$, and $\color{blue-violet}6_P$ edges also restricted to slots of their respective colours.
We will sometimes represent this by underlining elements or gaps in a schedule, \eg, 
[$\color{red}\SLOT$,
$\color{blue}\SLOT$,
$\color{applegreen}\SLOT$,
$\color{red}\SLOT$,
$\color{blue}\SLOT$,
$\color{blue-violet}\SLOT$], 
or by referring to edges as being red, blue, green, or purple.

All gadgets introduced below are constructed such that the lengths of their schedules are integer multiples of 6; in the final polycule $\mathcal P_{d\varphi}$ they will be connected (usually through intermediaries) to the True Clock, such that their edges must stick to certain slots.
\newcommand\slotgood{slot-respecting\xspace}%
To keep correctness proofs of individual gadgets readable, we call a schedule $S$
that schedules all coloured edges in slots of the given colour \emph{\slotgood}.

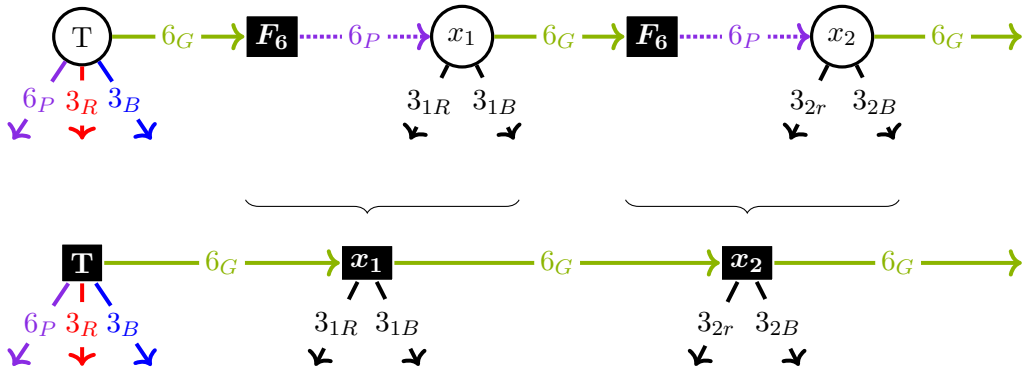
\begin{figure}[hbt] %
    \centering
    \begin{tikzpicture}[scale=2]
        \def\nodelayer{2}
        \def\nodeoutlayer{1.25}
        \def\bracelayer{1.25}
        \def\gadgetlayer{0.5}
        \def\gadgetoutlayer{-0.25}
        \def\Txpos{0}
        \def\xscon{1.25}

        \node[graph node] (T)                   at (\Txpos,                 \nodelayer) {T};
        \node[invisible node] (low left 1)      at (\Txpos-0.5,             \nodeoutlayer) {};
        \node[invisible node] (low left 2)      at (\Txpos,                 \nodeoutlayer) {};
        \node[invisible node] (low left 3)      at (\Txpos+0.5,             \nodeoutlayer) {};
        \node[shorthand] (connect T)            at (\Txpos+\xscon,          \nodelayer) {$F_6$};
        \node[graph node] (x1)                  at (\Txpos+2*\xscon,        \nodelayer){$x_1$};
        \node[invisible node] (low mid 1)       at (\Txpos+2*\xscon-0.375,  \nodeoutlayer) {};
        \node[invisible node] (low mid 2)       at (\Txpos+2*\xscon+0.375,  \nodeoutlayer) {};
        \node[shorthand] (connect x1)           at (\Txpos+3*\xscon,        \nodelayer) {$F_6$};
        \node[graph node] (x2)                  at (\Txpos+4*\xscon,        \nodelayer){$x_2$};
        \node[invisible node] (low right 1)     at (\Txpos+4*\xscon-0.375,  \nodeoutlayer) {};
        \node[invisible node] (low right 2)     at (\Txpos+4*\xscon+0.375,  \nodeoutlayer) {};
        \node[invisible node] (right)           at (\Txpos+5*\xscon,        \nodelayer) {};

        \path[->] (T) edge[std, ourGreen] node[edge descriptor] {\color{applegreen}$6_G$} (connect T);
        \path[->] (connect T) edge[std B, ourPurple] node[edge descriptor] {\color{blue-violet}$6_P$} (x1);
        \path[->] (T) edge[std, ourPurple] node[edge descriptor] {\color{blue-violet}$6_P$} (low left 1);
        \path[->] (T) edge[std, ourRed] node[edge descriptor] {\color{red}$3_R$} (low left 2);
        \path[->] (T) edge[std, ourBlue] node[edge descriptor] {\color{blue}$3_B$} (low left 3);
        \path[->] (x1) edge[std, ourBlack] node[edge descriptor] {$3_{1R}$} (low mid 1);
        \path[->] (x1) edge[std, ourBlack] node[edge descriptor] {$3_{1B}$} (low mid 2);
        \path[->] (x1) edge[std, ourGreen] node[edge descriptor] {\color{applegreen}$6_G$} (connect x1);
        \path[->] (connect x1) edge[std B, ourPurple] node[edge descriptor] {\color{blue-violet}$6_P$} (x2);
        \path[->] (x2) edge[std, ourBlack] node[edge descriptor] {$3_{2r}$} (low right 1);
        \path[->] (x2) edge[std, ourBlack] node[edge descriptor] {$3_{2B}$} (low right 2);
        \path[->] (x2) edge[std, ourGreen] node[edge descriptor] {\color{applegreen}$6_G$} (right);

        \draw [decorate,decoration={brace,amplitude=5pt,mirror,raise=4ex,aspect=0.445}]
        (\Txpos+\xscon-0.175,\bracelayer) -- 
        (\Txpos+2*\xscon+0.375,\bracelayer);

        \draw [decorate,decoration={brace,amplitude=5pt,mirror,raise=4ex,aspect=0.445}]
        (\Txpos+3*\xscon-0.175,\bracelayer) -- 
        (\Txpos+4*\xscon+0.375,\bracelayer);

        \node[shorthand] (T short)                  at (\Txpos,                 \gadgetlayer) {T};
        \node[invisible node] (short left 1)        at (\Txpos-0.5,             \gadgetoutlayer) {};
        \node[invisible node] (short left 3)        at (\Txpos,                 \gadgetoutlayer) {}; 
        \node[invisible node] (short left 4)        at (\Txpos+0.5,             \gadgetoutlayer) {};
        
        \node[shorthand] (x1 short)                 at (\Txpos+1.5*\xscon,      \gadgetlayer) {$x_1$};
        \node[invisible node] (short right low 1)   at (\Txpos+1.5*\xscon-0.375,\gadgetoutlayer) {};
        \node[invisible node] (short right low 2)   at (\Txpos+1.5*\xscon+0.375,\gadgetoutlayer) {};
        
        \node[shorthand] (x2 short)                 at (\Txpos+3.5*\xscon,      \gadgetlayer) {$x_2$};
        \node[invisible node] (short2 right right)  at (\Txpos+5*\xscon,      \gadgetlayer) {};
        \node[invisible node] (short2 right low 1)  at (\Txpos+3.5*\xscon-0.375,\gadgetoutlayer) {};
        \node[invisible node] (short2 right low 2)  at (\Txpos+3.5*\xscon+0.375,\gadgetoutlayer) {};

        \path[->] (T short) edge[std, ourPurple] node[edge descriptor] {\color{blue-violet}$6_P$} (short left 1);
        \path[->] (T short) edge[std, ourRed] node[edge descriptor] {\color{red}$3_R$} (short left 3);
        \path[->] (T short) edge[std, ourBlue] node[edge descriptor] {\color{blue}$3_B$} (short left 4);
        \path[->]  (T short) edge[std, ourGreen] node[edge descriptor] {\color{applegreen}$6_G$} (x1 short);

        \path[->]  (x1 short) edge[std, ourGreen] node[edge descriptor] {\color{applegreen}$6_G$} (x2 short);
        \path[->] (x1 short) edge[std, ourBlack] node[edge descriptor] {$3_{1R}$} (short right low 1);
        \path[->] (x1 short) edge[std, ourBlack] node[edge descriptor] {$3_{1B}$} (short right low 2);
        
        \path[->] (x2 short) edge[std, ourGreen] node[edge descriptor] {\color{applegreen}$6_G$} (short2 right right);
        \path[->] (x2 short) edge[std, ourBlack] node[edge descriptor] {$3_{2r}$} (short2 right low 1);
        \path[->] (x2 short) edge[std, ourBlack] node[edge descriptor] {$3_{2B}$} (short2 right low 2);

    \end{tikzpicture}
    \caption{
            Gadgets for the True Clock and for sample variables $x_1$ and $x_2$ (top, left to right), with shorthand versions shown below. Gadgets are shown connected as they would be in a sample variable layer, and their colours and schedules are discussed in 
            \wref{sec:TrueClock}. Further variables can be added to the right.
    }
    \label{fig:variables}
\end{figure} %

\paragraph{Drawing Conventions}
\label{para:in/out}
Gadgets are connected by input and output edges, represented by incoming and outgoing arrows respectively -- each one being half of a relationship between two people from different gadgets. 
In addition to their frequencies, input and output edges share restrictions on their permissible schedules, carrying them from one gadget to another as discussed in the proofs associated with each gadget.

In shorthand gadgets, vertical incoming edges are the primary \emph{input} to a gadget, encoding the value of some variable or logical function; horizontal inputs contain edges of fixed colour, which we will refer to as \emph{constants}. 
Similarly, vertical outgoing edges represent primary \emph{outputs} that encode the result of the gadget, while horizontal outgoing edges represent incidentally created constants which may either be used by other gadgets or connected to pendent vertices.

Some incoming and outgoing edges will have end labels of the form ``$\cdot i$'', indicating $i$ connections of the given type, each with a different person.

To show that $\mathcal P_{d\varphi}$ is bipartite, we assign a sex to each node: male nodes with solid edges, and female nodes with dotted edges. We will also use a directed graph to show the sex of the originating node: edges coming from males have solid edges, while edges coming from female nodes will have dotted edges. Note that this convention is only used to show that the complete graph is bipartite while reasoning about parts of it; in all other respects, the graph is undirected (as we would expect for poly scheduling).

\subsection{Variables}
\label{sec:variables}

\wref{fig:variables} also introduces the gadget for a sample variable, $x_1$, which again has four relationships:
[$3_{1R}, 3_{1B},$
$\color{applegreen}6_G$,
$\color{blue-violet}6_P$].
The key property of variable gadgets is summarized in the following Lemma.

\begin{lemma}[Variable gadget schedules]
\label{lem:variable-6-color}
    The $x_i$ node in each variable gadget $x_i$ has a valid local schedule which must be of the form
    $[3_a, 3_b, {\color{applegreen}6_G}, 
    3_a, 3_b, {\color{blue-violet}6_P}]$ in any \slotgood global schedule.
\end{lemma}

\begin{proof}
    $x_i$ nodes have tasks 
    [$3_{1R}, 3_{1B},$
    $\color{applegreen}6_G$,
    $\color{blue-violet}6_P$] 
    and local density $D=1$, so $3_{iR}$ and $3_{iB}$ must be scheduled exactly once in each 3-day period, forcing every 3 days to be of the form $[3_a, 3_b, \SLOT]$, and every 6-day schedule to be of the form $[3_a, 3_b, \SLOT, 3_a, 3_b, \SLOT]$; this leaves two remaining slots, which must contain $\color{applegreen}6_G$ and $\color{blue-violet}6_P$. 
    
    The $\color{applegreen}6_G$ edge of the first variable node, $x_1$, is shared with $T$ such that it must be green, so valid schedules for $x_1$
    are of the form
    $[3_a, 3_b, \color{applegreen}6_G\color{black}, 3_a, 3_b, \color{blue-violet}\SLOT\color{black}]$, which must be completed as $[3_a, 3_b, \color{applegreen}6_G\color{black}, 3_a, 3_b, \color{blue-violet}6_P\color{black}]$ -- a valid schedule for $x_1$.

    The $\color{blue-violet}6_P$ edge leaving $x_1$ is then connected to an $F_6$ flipper gadget, which returns a $\color{applegreen}6_G$ edge by \wref{lem:flipper6}.
    Further variables are each connected to the $\color{applegreen}6_G$ edge returned from the $F_6$ flipper gadget of the previous variable gadget, so their schedules must be of the same form.
    It therefore follows inductively that each variable gadget $x_i$ has a valid schedule, and that it must be of the form
    $[3_a, 3_b, \color{applegreen}6_G\color{black}, 3_a,$
    $ 3_b, \color{blue-violet}6_P\color{black}]$.
\end{proof}

According to \wref{lem:variable-6-color}, the incident $\color{applegreen}6_G$ edge is used to restrict the valid schedules for $x_1$, leaving two possibilities: 
$[
    \mathunderline{red}{3_{1R}},
    \mathunderline{blue}{3_{1B}},
    \color{applegreen}6_G\color{black},
    \mathunderline{red}{3_{1R}},
    \mathunderline{blue}{3_{1B}},
    \color{blue-violet}6_P\color{black}
]$ 
and
$[
    \mathunderline{red}{3_{1B}},
    \mathunderline{blue}{3_{1R}},
    \color{applegreen}6_G\color{black},
    \mathunderline{red}{3_{1B}},
    \mathunderline{blue}{3_{1R}},
    \color{blue-violet}6_P\color{black}
]$. 
The former schedule, where $3_{1R}$ is scheduled in red slots, corresponds to a variable assignment where $x_1$ is True, whereas the the second schedule corresponds to $x_1$ being assigned False.

This technique of connecting $\color{red}3_R$, $ \color{blue}3_B$, $\color{applegreen}6_G$, or $\color{blue-violet}6_P$ edges to nodes in a gadget in order to limit their valid local schedules and force relationships between edges, slots, and particular meanings will be used extensively in what follows. 

Note that $3_{iB}$ has the opposite value to $3_{iR}$, so using $3_{iB}$ edges in the polycule corresponds to the negated literal $\overline{x_i}$, just as $3_{iR}$ edges correspond to the literal $x_i$.

\subsection{Flippers}
\label{sec:flipper}

Flippers are introduced by \wref{fig:flipper}, and are used to ensure that  $\mathcal P_{d\varphi}=(P,R,f)$ is bipartite by converting between male and female edges. They allow variable and duplicator gadgets to consume a constant, rather than linear, amount of $\color{blue}3_b$, $\color{applegreen}6_G$, and $\color{blue-violet}6_P$ edges.

\begin{lemma}[F-6 Flipper gadget schedules]
\label{lem:flipper6}
    The $F_6$ node in each $F_6$ flipper gadget has a valid local schedule which must be of the form 
    $[3_a, 3_b, {\color{applegreen}6_G}, 
    3_a, 3_b, {\color{blue-violet}6_P}]$ in any \slotgood global schedule.
\end{lemma}
\begin{proof}
    Note that $F_6$ flippers and variables share the same tasks: $[3,3,6,6]$. By the proof of \wref{lem:variable-6-color}, schedules for $F_6$ flippers must be of the form $[3_a, 3_b, \SLOT, 3_a, 3_b, \SLOT]$, where the remaining slots must contain $\color{applegreen}6_G$ and $\color{blue-violet}6_P$ edges. The incoming edge, either a $\color{applegreen}6_G$ or $\color{blue-violet}6_P$, must be in a corresponding slot forcing schedules to be of the form $[3_a, 3_b, \color{applegreen}6_G\color{black}, 3_a, 3_b, \color{blue-violet}\SLOT\color{black}]$
    or
    $[3_a, 3_b, \color{applegreen}\SLOT\color{black}, 3_a, 3_b, \color{blue-violet}6_P\color{black}]$ respectively. Both must yield the complete and valid local schedule $[3_a, 3_b, {\color{applegreen}6_G}, 
    3_a, 3_b, {\color{blue-violet}6_P}]$.
\end{proof}

\begin{lemma}[F-3 Flipper gadget schedules]
\label{lem:flipper3}
    The $F_3$ node in each $F_3$ flipper gadget has a valid local schedule which must be of the form  
$[{\color{red}3_R},
{\color{blue}3_B},
{\color{applegreen}6_G},
{\color{red}3_R},
{\color{blue}3_B},
{\color{blue-violet}6_P]}$ in any \slotgood global schedule.
\end{lemma}

\begin{proof}
    $F_3$ flipper nodes have two inputs, [$3_i, 6_i$] and two outputs, [$3_o, 6_o$]. 
    $3_i$ and $6_i$ are indirectly connected to the True Clock such that $3_i$ must be red or blue, and $6_i$ must be either green or purple.
    This forces partial schedules to be of the form 
    [$3_i,
    \SLOT,
    6_i,
    3_i,
    \SLOT,
    \SLOT$]
    or
    [$3_i,
    6_i,
    \SLOT,
    3_i,
    \SLOT,
    \SLOT$]. 
    As $F_3$ nodes have density $D=1$, the $6_o$ edges must be placed 3 days after the $6_i$ edge to avoid violating the condition of the $3_o$ edge, for schedules of the form
    [$3_i,
    3_o,
    6_i,
    3_i,
    3_o,
    6_o$]
    and
    [$3_i,
    6_i,
    3_o,
    3_i,
    6_o,
    3_o$] respectively. 
    If $3_i$ is red, $3_o$ will be blue, and vice versa; similarly, one 6 edge will be purple and the other green.
    In either case, the resultant schedule will be valid and of the form
    [$\color{red}3_R$,
    $\color{blue}3_B$,
    $\color{applegreen}6_G$,
    $\color{red}3_R$,
    $\color{blue}3_B$,
    $\color{blue-violet}6_P$].
\end{proof}

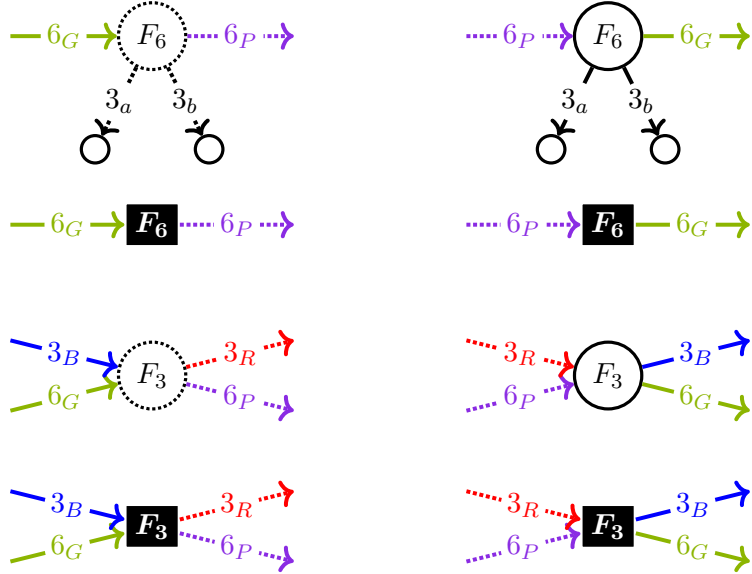
\begin{figure}[hbt] 
    \centering
    \begin{tikzpicture}[scale=2]
        \def\nodelayer{2}
        \def\nodeoutlayer{1.25}
        \def\bracelayer{1.25}
        \def\gadgetlayer{0.75}
        \def\Txpos{0}
        \def\xscon{1}

        \node[invisible node] (T)               at (\Txpos,                 \nodelayer) {};
        \node[graph node B] (connect T)         at (\Txpos+\xscon,          \nodelayer) {$F_6$};
        \node[graph node] (fake out left1)      at (\Txpos+\xscon-0.375,    \nodeoutlayer) {};
        \node[graph node] (fake out right1)     at (\Txpos+\xscon+0.375,    \nodeoutlayer) {};
        \node[invisible node] (x1 left)         at (\Txpos+2*\xscon,        \nodelayer){};
        \node[invisible node] (x1 right)        at (\Txpos+3*\xscon,        \nodelayer){};
        \node[graph node] (connect x1)          at (\Txpos+4*\xscon,          \nodelayer) {$F_6$};
        \node[graph node] (fake out left2)      at (\Txpos+4*\xscon-0.375,  \nodeoutlayer) {};
        \node[graph node] (fake out right2)     at (\Txpos+4*\xscon+0.375,  \nodeoutlayer) {};
        \node[invisible node] (x2)              at (\Txpos+5*\xscon,        \nodelayer){};

        \path[->] (T) edge[std, ourGreen] node[edge descriptor] {\color{applegreen}$6_G$} (connect T);
        \path[->] (connect T) edge[std B, ourPurple] node[edge descriptor] {\color{blue-violet}$6_P$} (x1 left);
        \path[->] (connect T) edge[std B, ourBlack] node[edge descriptor] {$3_a$} (fake out left1);
        \path[->] (connect T) edge[std B, ourBlack] node[edge descriptor] {$3_b$} (fake out right1);
        \path[->] (x1 right) edge[std B, ourPurple] node[edge descriptor] {\color{blue-violet}$6_P$} (connect x1);
        \path[->] (connect x1) edge[std, ourGreen] node[edge descriptor] {\color{applegreen}$6_G$} (x2);
        \path[->] (connect x1) edge[std, ourBlack] node[edge descriptor] {$3_a$} (fake out left2);
        \path[->] (connect x1) edge[std, ourBlack] node[edge descriptor] {$3_b$} (fake out right2);

        \node[invisible node] (T short)         at (\Txpos,             \gadgetlayer) {};
        \node[shorthand] (F left)               at (\Txpos+\xscon,      \gadgetlayer) {$F_6$};
        \node[invisible node] (x1 short left)   at (\Txpos+2*\xscon,    \gadgetlayer) {};
        \node[invisible node] (x1 short right)  at (\Txpos+3*\xscon,    \gadgetlayer) {};
        \node[shorthand] (F right)              at (\Txpos+4*\xscon,    \gadgetlayer) {$F_6$};
        \node[invisible node] (x2 short)        at (\Txpos+5*\xscon,    \gadgetlayer) {};

        \path[->] (T short) edge[std, ourGreen] node[edge descriptor] {\color{applegreen}$6_G$} (F left);
        \path[->] (F left) edge[std B, ourPurple] node[edge descriptor] {\color{blue-violet}$6_P$} (x1 short left);
        \path[->] (x1 short right) edge[std B, ourPurple] node[edge descriptor] {\color{blue-violet}$6_P$} (F right);
        \path[->] (F right) edge[std, ourGreen] node[edge descriptor] {\color{applegreen}$6_G$} (x2 short);

        \def\nodelayer{1.75-2}
        \def\gadgetlayer{0.75-2}

        \def\Txpos{0}
        \def\xscon{1}

        \node[invisible node] (T)               at (\Txpos,                 \nodelayer-0.25) {};
        \node[invisible node] (T2)              at (\Txpos,                 \nodelayer+0.25) {};
        \node[graph node B] (connect T)         at (\Txpos+\xscon,          \nodelayer) {$F_3$};
        \node[invisible node] (x1 left)         at (\Txpos+2*\xscon,        \nodelayer-0.25){};
        \node[invisible node] (x1 left2)        at (\Txpos+2*\xscon,        \nodelayer+0.25){};
        \node[invisible node] (x1 right)        at (\Txpos+3*\xscon,        \nodelayer-0.25){};
        \node[invisible node] (x1 right2)        at (\Txpos+3*\xscon,       \nodelayer+0.25){};
        \node[graph node] (connect x1)          at (\Txpos+4*\xscon,        \nodelayer) {$F_3$};
        \node[invisible node] (x2)              at (\Txpos+5*\xscon,        \nodelayer-0.25){};
        \node[invisible node] (x22)              at (\Txpos+5*\xscon,       \nodelayer+0.25){};

        \path[->] (T) edge[std, ourGreen] node[edge descriptor] {\color{applegreen}$6_G$} (connect T);
        \path[->] (T2) edge[std, ourBlue] node[edge descriptor] {\color{blue}$3_B$} (connect T);
        \path[->] (connect T) edge[std B, ourPurple] node[edge descriptor] {\color{blue-violet}$6_P$} (x1 left);
        \path[->] (connect T) edge[std B, ourRed] node[edge descriptor] {\color{red}$3_R$} (x1 left2);
        \path[->] (x1 right) edge[std B, ourPurple] node[edge descriptor] {\color{blue-violet}$6_P$} (connect x1);
        \path[->] (x1 right2) edge[std B, ourRed] node[edge descriptor] {\color{red}$3_R$} (connect x1);
        \path[->] (connect x1) edge[std, ourGreen] node[edge descriptor] {\color{applegreen}$6_G$} (x2);
        \path[->] (connect x1) edge[std, ourBlue] node[edge descriptor] {\color{blue}$3_B$} (x22);

        \node[invisible node] (T short)         at (\Txpos,             \gadgetlayer-0.25) {};
        \node[invisible node] (T short2)        at (\Txpos,             \gadgetlayer+0.25) {};        
        \node[shorthand] (F left)               at (\Txpos+\xscon,      \gadgetlayer) {$F_3$};
        \node[invisible node] (x1 short left)   at (\Txpos+2*\xscon,    \gadgetlayer-0.25) {};
        \node[invisible node] (x1 short left2)  at (\Txpos+2*\xscon,    \gadgetlayer+0.25) {};        
        \node[invisible node] (x1 short right)  at (\Txpos+3*\xscon,    \gadgetlayer-0.25) {};
        \node[invisible node] (x1 short right2) at (\Txpos+3*\xscon,    \gadgetlayer+0.25) {};        
        \node[shorthand] (F right)              at (\Txpos+4*\xscon,    \gadgetlayer) {$F_3$};
        \node[invisible node] (x2 short)        at (\Txpos+5*\xscon,    \gadgetlayer-0.25) {};
        \node[invisible node] (x2 short2)       at (\Txpos+5*\xscon,    \gadgetlayer+0.25) {};

        \path[->] (T short) edge[std, ourGreen] node[edge descriptor] {\color{applegreen}$6_G$} (F left);
        \path[->] (T short2) edge[std, ourBlue] node[edge descriptor] {\color{blue}$3_B$} (F left);
        \path[->] (F left) edge[std B, ourPurple] node[edge descriptor] {\color{blue-violet}$6_P$} (x1 short left);
        \path[->] (F left) edge[std B, ourRed] node[edge descriptor] {\color{red}$3_R$} (x1 short left2);
        
        \path[->] (x1 short right) edge[std B, ourPurple] node[edge descriptor] {\color{blue-violet}$6_P$} (F right);
        \path[->] (x1 short right2) edge[std B, ourRed] node[edge descriptor] {\color{red}$3_R$} (F right);
        \path[->] (F right) edge[std, ourGreen] node[edge descriptor] {\color{applegreen}$6_G$} (x2 short);
        \path[->] (F right) edge[std, ourBlue] node[edge descriptor] {\color{blue}$3_B$} (x2 short2);
        
    \end{tikzpicture}
    \caption{
            Four sample flipper gadgets, each with shorthand versions below them: $\color{applegreen}6_G$ male to $\color{blue-violet}6_P$ female (top left), $\color{blue-violet}6_P$ female to $\color{applegreen}6_G$ male (top right), $\color{blue}3_B$ male to $\color{red}3_R$ female (bottom left), and $\color{red}3_R$ female to $\color{blue}3_B$ male (bottom right). Note that $F_3$ flippers act incidentally as $F_6$ flippers, but are used separately in our construction.
    }
    \label{fig:flipper}
\end{figure} 

\subsection{Duplication of Variables and Constants}
\label{sec:duplication}

Variables may appear in multiple clauses, while the constant $\color{red}3_R$, $\color{blue}3_B$, $\color{applegreen}6_G$, and $\color{blue-violet}6_P$ edges are used in multiple gadgets, engendering a need for the duplication of variables and constants. 

\subsubsection{3-Duplicators}
\label{sec:D3}

A gadget for duplicating edges with period 3 is shown in \wref{fig:3-splitter} and proven to accurately reproduce 
its input by \wref{lem:D3}.

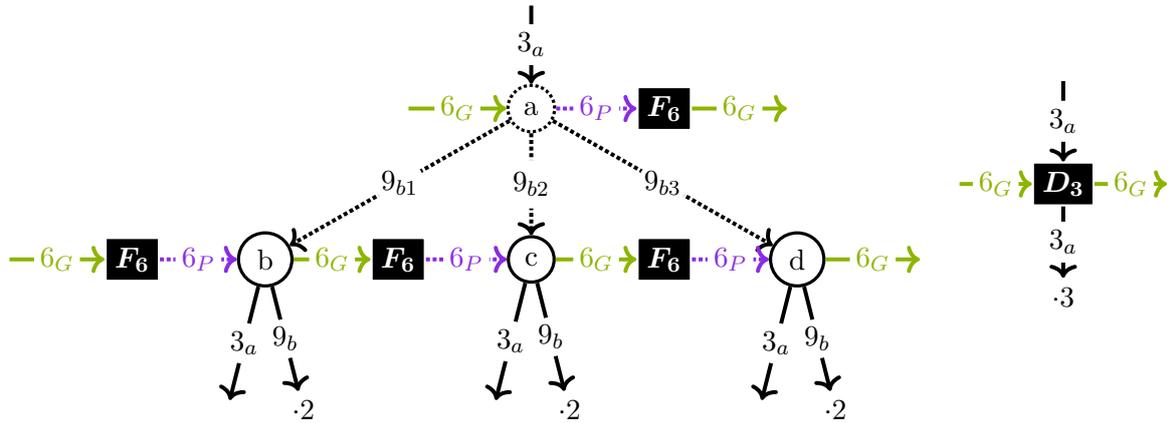
\begin{figure}[bth] 
    \centering
    \begin{tikzpicture}[scale=2]
        \def\LtwoY{-1}
        \def\LoutY{-2}
        \def\LtwoXsep{1.75}
        
        \node[invisible node] (top) at (0, 0.75) {};
        
        \node[invisible node] (top in) at (-0.5*\LtwoXsep, 0) {};
        \node[shorthand] (top short) at (0.5*\LtwoXsep, 0) {$F_6$};
        \node[invisible node] (top out) at (\LtwoXsep, 0) {};
        \node[graph node B] (D1) at (0, 0) {a};

        \path[->] (top) edge[std, ourBlack] node[edge descriptor] {$3_a$} (D1);
        \path[->] (top in) edge[std, ourGreen] node[edge descriptor] {\color{applegreen}$6_G$} (D1);
        \path[->] (D1) edge[std B, ourPurple] node[edge descriptor] {\color{blue-violet}$6_P$} (top short);
        \path[->] (top short) edge[std, ourGreen] node[edge descriptor] {\color{applegreen}$6_G$} (top out);

        \node[invisible node] (in) at (-2*\LtwoXsep, \LtwoY) {};
        \node[shorthand] (flipper b)  at (-1.5*\LtwoXsep, \LtwoY) {$F_6$};        
        \node[graph node] (D2) at (-1*\LtwoXsep, \LtwoY) {b};
        \node[shorthand] (flipper c)  at (-0.5*\LtwoXsep, \LtwoY) {$F_6$};
        \node[graph node] (D3) at (0, \LtwoY) {c};
        \node[shorthand] (flipper d)  at (0.5*\LtwoXsep, \LtwoY) {$F_6$};
        \node[graph node] (D4) at (\LtwoXsep, \LtwoY) {d};
        \node[invisible node] (out)  at (1.5*\LtwoXsep, \LtwoY) {};

        \path[->] (D1) edge[std B, ourBlack] node[edge descriptor] {$9_{b1}$} (D2);
        \path[->] (D1) edge[std B, ourBlack] node[edge descriptor] {$9_{b2}$} (D3);
        \path[->] (D1) edge[std B, ourBlack] node[edge descriptor] {$9_{b3}$} (D4);

        \path[->] (flipper b) edge[std B, ourPurple] node[edge descriptor] {\color{blue-violet}$6_P$} (D2);
        \path[->] (flipper c) edge[std B, ourPurple] node[edge descriptor] {\color{blue-violet}$6_P$} (D3);
        \path[->] (flipper d) edge[std B, ourPurple] node[edge descriptor] {\color{blue-violet}$6_P$} (D4);

        \path[->] (in) edge[std, ourGreen] node[edge descriptor] {\color{applegreen}$6_G$} (flipper b);
        \path[->] (D2) edge[std, ourGreen] node[edge descriptor] {\color{applegreen}$6_G$} (flipper c);
        \path[->] (D3) edge[std, ourGreen] node[edge descriptor] {\color{applegreen}$6_G$} (flipper d);
        \path[->] (D4) edge[std, ourGreen] node[edge descriptor] {\color{applegreen}$6_G$} (out);

        \node[invisible node] (output 1) at (-1*\LtwoXsep-0.25, \LoutY) {};
        \node[invisible node] (repeat 1) at (-1*\LtwoXsep+0.25, \LoutY) {$\cdot2$};
        
        \node[invisible node] (output 2) at (-0.25, \LoutY) {};
        \node[invisible node] (repeat 2) at (0.25, \LoutY) {$\cdot2$};
        
        \node[invisible node] (output 3) at (1*\LtwoXsep-0.25, \LoutY) {};
        \node[invisible node] (repeat 3) at (1*\LtwoXsep+0.25, \LoutY) {$\cdot2$};

        \path[->] (D2) edge[std, ourBlack] node[edge descriptor] {$3_a$} (output 1);
        \path[->] (D2) edge[std, ourBlack] node[edge descriptor] {$9_b$} (repeat 1);
         
        \path[->] (D3) edge[std, ourBlack] node[edge descriptor] {$3_a$} (output 2);
        \path[->] (D3) edge[std, ourBlack] node[edge descriptor] {$9_b$} (repeat 2);
        
        \path[->] (D4) edge[std, ourBlack] node[edge descriptor] {$3_a$} (output 3);
        \path[->] (D4) edge[std, ourBlack] node[edge descriptor] {$9_b$} (repeat 3);

        \def\yadj{2}
        \def\xadj{1}

        \node[shorthand] (D3 short)             at (\xadj+2.5, -2.5+\yadj) {$D_{3}$};
        \node[invisible node] (short input a)   at (\xadj+2.5, -1.75+\yadj) {};
        \node[invisible node] (short input G)   at (\xadj+1.75, -2.5+\yadj) {};
        \node[invisible node] (short output a)  at (\xadj+2.5, -3.25+\yadj) {$\cdot 3$};
        \node[invisible node] (short output G)  at (\xadj+3.25, -2.5+\yadj) {};    
        
        \path[->] (short input a) edge[std, ourBlack] node[edge descriptor] {$3_a$} (D3 short);
        \path[->] (short input G) edge[std, ourGreen] node[edge descriptor] {\color{applegreen}$6_G$} (D3 short);
        \path[->] (D3 short) edge[std, ourBlack] node[edge descriptor] {$3_a$} (short output a);
        \path[->] (D3 short) edge[std, ourGreen] node[edge descriptor] {\color{applegreen}$6_G$} (short output G);
        
    \end{tikzpicture}
\caption{
A gadget for duplicating input edges with frequency 3, with a shorthand version on the right. Note that the input can be duplicated indefinitely many times by adding further layers which each use the $9_{b}$ edges from the previous layer, but the output will alternate between male and female $3_a$ edges. Odd layers will have female main nodes with $F_6$ flipper nodes on the right, as shown in the top layer, while even layers will have male main nodes with $F_6$ flipper nodes on the left, as shown on the bottom layer. $\color{applegreen}6_P$ edges can be connected continuously from layer to layer, with the final edge released to another gadget.
}
\label{fig:3-splitter}
\end{figure} 

\begin{lemma}[3-Duplicator gadget schedules]
    \label{lem:D3}
    The non-flipper nodes in each 3-duplicator gadget $D_3$ have valid local schedules which must be of the form of either
    \begin{align*}
    &[
        \mathunderline{red}{3_a},
        \mathunderline{blue}{9_{b}},
        \color{applegreen}6_G\color{black},
        \mathunderline{red}{3_a},
        \mathunderline{blue}{9_{b'}},
        \color{blue-violet}6_P\color{black},
        \mathunderline{red}{3_a},
        \mathunderline{blue}{9_{b''}},
        \color{applegreen}6_G\color{black},
        \mathunderline{red}{3_a},
        \mathunderline{blue}{9_{b}},
        \color{blue-violet}6_P\color{black},
        \mathunderline{red}{3_a},
        \mathunderline{blue}{9_{b'}},
        \color{applegreen}6_G\color{black},
        \mathunderline{red}{3_a},
        \mathunderline{blue}{9_{b''}},
        \color{blue-violet}6_P\color{black}
    ]
    \ \text{or}
    \\
    &[
        \mathunderline{red}{9_{b}},
        \mathunderline{blue}{3_a},
        \color{applegreen}6_G\color{black},
        \mathunderline{red}{9_{b'}},
        \mathunderline{blue}{3_a},
        \color{blue-violet}6_P\color{black},
        \mathunderline{red}{9_{b''}},
        \mathunderline{blue}{3_a},
        \color{applegreen}6_G\color{black},
        \mathunderline{red}{9_{b}},
        \mathunderline{blue}{3_a},
        \color{blue-violet}6_P\color{black},
        \mathunderline{red}{9_{b'}},
        \mathunderline{blue}{3_a},
        \color{applegreen}6_G\color{black},
        \mathunderline{red}{9_{b''}},
        \mathunderline{blue}{3_a},
        \color{blue-violet}6_P\color{black}
    ]
    \end{align*}
    in any \slotgood global schedule.
\end{lemma}
\begin{proof}
    Each non-flipper node in $D_3$ has tasks 
    [$3_a,$
    $\color{blue-violet}6_P$, $\color{applegreen}6_G$,
    $9_{b1}, 9_{b2}, 9_{b3}$]
    and has local
    density $D=1$, so each task with frequency $f$ must appear exactly once every $f$ days. 
    Further, in any \slotgood schedule, the $\color{applegreen}6_G$ edge must be scheduled in green slots, forcing partial schedules of the form
        $[{\color{red}\SLOT},
        {\color{blue}\SLOT},
        {\color{applegreen}6_G},
        {\color{red}\SLOT},
        {\color{blue}\SLOT},
        {\color{blue-violet}\SLOT}]$.

    Considering node $a$, note that if incident edge $3_a$ is scheduled on a combination of red and blue days then either it will appear more than once in some 3-day period or its constraint will be violated.
    This demonstrates that schedules must be of the form
    \begin{align*}
    &[\mathunderline{red}{3_a},
        \color{blue}\SLOT\color{black},
        \color{applegreen}6_G\color{black},
        \mathunderline{red}{3_a},
        \color{blue}\SLOT\color{black},
        \color{blue-violet}\SLOT\color{black},
        \mathunderline{red}{3_a},
        \color{blue}\SLOT\color{black},
        \color{applegreen}6_G\color{black},
        \mathunderline{red}{3_a},
        \color{blue}\SLOT\color{black},
        \color{blue-violet}\SLOT\color{black},
        \mathunderline{red}{3_a},
        \color{blue}\SLOT\color{black},
        \color{applegreen}6_G\color{black},
        \mathunderline{red}{3_a},
        \color{blue}\SLOT\color{black},
        \color{blue-violet}\SLOT\color{black}]
        \ \text{or}
        \\
    &[\color{red}\SLOT\color{black},
        \mathunderline{blue}{3_a},
        \color{applegreen}6_G\color{black},
        \color{red}\SLOT\color{black},
        \mathunderline{blue}{3_a},
        \color{blue-violet}\SLOT\color{black},
        \color{red}\SLOT\color{black},
        \mathunderline{blue}{3_a},
        \color{applegreen}6_G\color{black},
        \color{red}\SLOT\color{black},
        \mathunderline{blue}{3_a},
        \color{blue-violet}\SLOT\color{black},
        \color{red}\SLOT\color{black},
        \mathunderline{blue}{3_a},
        \color{applegreen}6_G\color{black},
        \color{red}\SLOT\color{black},
        \mathunderline{blue}{3_a},
        \color{blue-violet}\SLOT\color{black}],
    \end{align*}
    depending on the colour of the incident $3_a$ edge.

    Assume towards a contradiction that the $\color{blue-violet}6_P$ edge is not consistently scheduled in purple slots. This produces the schedules:
    \begin{align*}
    &[\mathunderline{red}{3_a},
        \mathunderline{blue}{6_P},
        \color{applegreen}6_G\color{black},
        \mathunderline{red}{3_a},
        \color{blue}\SLOT\color{black},
        \color{blue-violet}\SLOT\color{black},
        \mathunderline{red}{3_a},
        \mathunderline{blue}{6_P},
        \color{applegreen}6_G\color{black},
        \mathunderline{red}{3_a},
        \color{blue}\SLOT\color{black},
        \color{blue-violet}\SLOT\color{black},
        \mathunderline{red}{3_a},
       \mathunderline{blue}{6_P},
        \color{applegreen}6_G\color{black},
        \mathunderline{red}{3_a},
        \color{blue}\SLOT\color{black},
        \color{blue-violet}\SLOT\color{black}]
        \ \text{and}
        \\
    &[\mathunderline{red}{6_P},
        \mathunderline{blue}{3_a},
        \color{applegreen}6_G\color{black},
        \color{red}\SLOT\color{black},
        \mathunderline{blue}{3_a},
        \color{blue-violet}\SLOT\color{black},
        \mathunderline{red}{6_P},
        \mathunderline{blue}{3_a},
        \color{applegreen}6_G\color{black},
        \color{red}\SLOT\color{black},
        \mathunderline{blue}{3_a},
        \color{blue-violet}\SLOT\color{black},
        \mathunderline{red}{6_P},
        \mathunderline{blue}{3_a},
        \color{applegreen}6_G\color{black},
        \color{red}\SLOT\color{black},
        \mathunderline{blue}{3_a},
        \color{blue-violet}\SLOT\color{black}],
    \end{align*}
    and their cyclic permutations. Both schedules have two free slots in the first 9 day period, and four free slots in the second 9 day period - forcing a contradiction when the $9_{b1}$, $9_{b2}$ and $9_{b3}$ edges are added.

    Thus, two partial schedules remain:
    \begin{align*}
    &[\mathunderline{red}{3_a},
        \mathunderline{blue}{6_P},
        \color{applegreen}6_G\color{black},
        \mathunderline{red}{3_a},
        \color{blue}\SLOT\color{black},
        \color{blue-violet}6_P\color{black},
        \mathunderline{red}{3_a},
        \mathunderline{blue}{6_P},
        \color{applegreen}6_G\color{black},
        \mathunderline{red}{3_a},
        \color{blue}\SLOT\color{black},
        \color{blue-violet}6_P\color{black},
        \mathunderline{red}{3_a},
       \mathunderline{blue}{6_P},
        \color{applegreen}6_G\color{black},
        \mathunderline{red}{3_a},
        \color{blue}\SLOT\color{black},
        \color{blue-violet}6_P\color{black}]
        \ \text{and}
        \\
    &[\mathunderline{red}{6_P},
        \mathunderline{blue}{3_a},
        \color{applegreen}6_G\color{black},
        \color{red}\SLOT\color{black},
        \mathunderline{blue}{3_a},
        \color{blue-violet}6_P\color{black},
        \mathunderline{red}{6_P},
        \mathunderline{blue}{3_a},
        \color{applegreen}6_G\color{black},
        \color{red}\SLOT\color{black},
        \mathunderline{blue}{3_a},
        \color{blue-violet}6_P\color{black},
        \mathunderline{red}{6_P},
        \mathunderline{blue}{3_a},
        \color{applegreen}6_G\color{black},
        \color{red}\SLOT\color{black},
        \mathunderline{blue}{3_a},
        \color{blue-violet}6_P\color{black}],
    \end{align*}
    In either case, $9_{b1}$, $9_{b2}$ and $9_{b3}$ must occupy the remaining slots, which match the schedules shown in the Lemma for some mapping of $9_{b1}$, $9_{b2}$ and $9_{b3}$ to $9_{b}$, $9_{b'}$ and $9_{b''}$.

    Now consider an arbitrary non-flipper node $p\neq a$, with an incident $9_{b}$ node. If the $3_a$ edge incident to $a$ is red, the schedule for $n$ must be of the form

    \begin{align*}
    &[\color{red}\SLOT\color{black},
        \mathunderline{blue}{9_{b}},
        \color{applegreen}6_G\color{black},
        \color{red}\SLOT\color{black},
        \color{blue}\SLOT\color{black},
        \color{blue-violet}\SLOT\color{black},
        \color{red}\SLOT\color{black},
        \color{blue}\SLOT\color{black},
        \color{applegreen}6_G\color{black},
        \color{red}\SLOT\color{black},
        \mathunderline{blue}{9_{b}},
        \color{blue-violet}\SLOT\color{black},
        \color{red}\SLOT\color{black},
        \color{blue}\SLOT\color{black},
        \color{applegreen}6_G\color{black},
        \color{red}\SLOT\color{black},
        \color{blue}\SLOT\color{black},
        \color{blue-violet}\SLOT\color{black}],
    \\
    &[\color{red}\SLOT\color{black},
        \color{blue}\SLOT\color{black},
        \color{applegreen}6_G\color{black},
        \color{red}\SLOT\color{black},
        \mathunderline{blue}{9_{b}},
        \color{blue-violet}\SLOT\color{black},
        \color{red}\SLOT\color{black},
        \color{blue}\SLOT\color{black},
        \color{applegreen}6_G\color{black},
        \color{red}\SLOT\color{black},
        \color{blue}\SLOT\color{black},
        \color{blue-violet}\SLOT\color{black},
        \color{red}\SLOT\color{black},
        \mathunderline{blue}{9_{b}},
        \color{applegreen}6_G\color{black},
        \color{red}\SLOT\color{black},
        \color{blue}\SLOT\color{black},
        \color{blue-violet}\SLOT\color{black}],\ \text{or}
    \\
    &[\color{red}\SLOT\color{black},
        \color{blue}\SLOT\color{black},
        \color{applegreen}6_G\color{black},
        \color{red}\SLOT\color{black},
        \color{blue}\SLOT\color{black},
        \color{blue-violet}\SLOT\color{black},
        \color{red}\SLOT\color{black},
        \mathunderline{blue}{9_{b}},
        \color{applegreen}6_G\color{black},
        \color{red}\SLOT\color{black},
        \color{blue}\SLOT\color{black},
        \color{blue-violet}\SLOT\color{black},
        \color{red}\SLOT\color{black},
        \color{blue}\SLOT\color{black},
        \color{applegreen}6_G\color{black},
        \color{red}\SLOT\color{black},
        \mathunderline{blue}{9_{b}},
        \color{blue-violet}\SLOT\color{black}].
    \end{align*}
    In any case, if the $3_a$ edge is scheduled in either blue or purple slots, its constraint will be violated, so it must be scheduled in the remaining red slots, leading to the following partial schedules:
    \begin{align*}
    &[\mathunderline{red}{3_{a}},
        \mathunderline{blue}{9_{b}},
        \color{applegreen}6_G\color{black},
        \mathunderline{red}{3_{a}},
        \color{blue}\SLOT\color{black},
        \color{blue-violet}\SLOT\color{black},
        \mathunderline{red}{3_{a}},
        \color{blue}\SLOT\color{black},
        \color{applegreen}6_G\color{black},
        \mathunderline{red}{3_{a}},
        \mathunderline{blue}{9_{b}},
        \color{blue-violet}\SLOT\color{black},
        \mathunderline{red}{3_{a}},
        \color{blue}\SLOT\color{black},
        \color{applegreen}6_G\color{black},
        \mathunderline{red}{3_{a}},
        \color{blue}\SLOT\color{black},
        \color{blue-violet}\SLOT\color{black}],
    \\
    &[\mathunderline{red}{3_{a}},
        \color{blue}\SLOT\color{black},
        \color{applegreen}6_G\color{black},
        \mathunderline{red}{3_{a}},
        \mathunderline{blue}{9_{b}},
        \color{blue-violet}\SLOT\color{black},
        \mathunderline{red}{3_{a}},
        \color{blue}\SLOT\color{black},
        \color{applegreen}6_G\color{black},
        \mathunderline{red}{3_{a}},
        \color{blue}\SLOT\color{black},
        \color{blue-violet}\SLOT\color{black},
        \mathunderline{red}{3_{a}},
        \mathunderline{blue}{9_{b}},
        \color{applegreen}6_G\color{black},
        \mathunderline{red}{3_{a}},
        \color{blue}\SLOT\color{black},
        \color{blue-violet}\SLOT\color{black}],\ \text{or}
    \\
    &[\mathunderline{red}{3_{a}},
        \color{blue}\SLOT\color{black},
        \color{applegreen}6_G\color{black},
        \mathunderline{red}{3_{a}},
        \color{blue}\SLOT\color{black},
        \color{blue-violet}\SLOT\color{black},
        \mathunderline{red}{3_{a}},
        \mathunderline{blue}{9_{b}},
        \color{applegreen}6_G\color{black},
        \mathunderline{red}{3_{a}},
        \color{blue}\SLOT\color{black},
        \color{blue-violet}\SLOT\color{black},
        \mathunderline{red}{3_{a}},
        \color{blue}\SLOT\color{black},
        \color{applegreen}6_G\color{black},
        \mathunderline{red}{3_{a}},
        \mathunderline{blue}{9_{b}},
        \color{blue-violet}\SLOT\color{black}].
    \end{align*}
    If either of the $9_{b'}$ or $9_{b''}$ edges are scheduled in the remaining purple slots, they will conflict with the $\color{applegreen}6_G$ edge, so the $9_{b'}$ and $9_{b''}$ edges must be scheduled in the remaining blue slots, and the $\color{blue-violet}6_P$ edge must be purple.
    This leads to full schedules of the form 
    \[[
        \mathunderline{red}{3_a},
        \mathunderline{blue}{9_{b}},
        \color{applegreen}6_G\color{black},
        \mathunderline{red}{3_a},
        \mathunderline{blue}{9_{b'}},
        \color{blue-violet}6_P\color{black},
        \mathunderline{red}{3_a},
        \mathunderline{blue}{9_{b''}},
        \color{applegreen}6_G\color{black},
        \mathunderline{red}{3_a},
        \mathunderline{blue}{9_{b}},
        \color{blue-violet}6_P\color{black},
        \mathunderline{red}{3_a},
        \mathunderline{blue}{9_{b'}},
        \color{applegreen}6_G\color{black},
        \mathunderline{red}{3_a},
        \mathunderline{blue}{9_{b''}},
        \color{blue-violet}6_P\color{black}
    ],\] 
    which again matches the Lemma.
    If the $3_a$ edge incident to $a$ is blue instead, the same logic applies, with all $3_a$ edges also being blue and all $9_b$ nodes being red.
\end{proof}

\subsubsection{6-Duplicators}
\label{sec:D6}
\wref{fig:6-splitter} and \wref{lem:D6} introduce a gadget which duplicates incident $\color{applegreen}6_G$ or $\color{blue-violet}6_P$ edges.

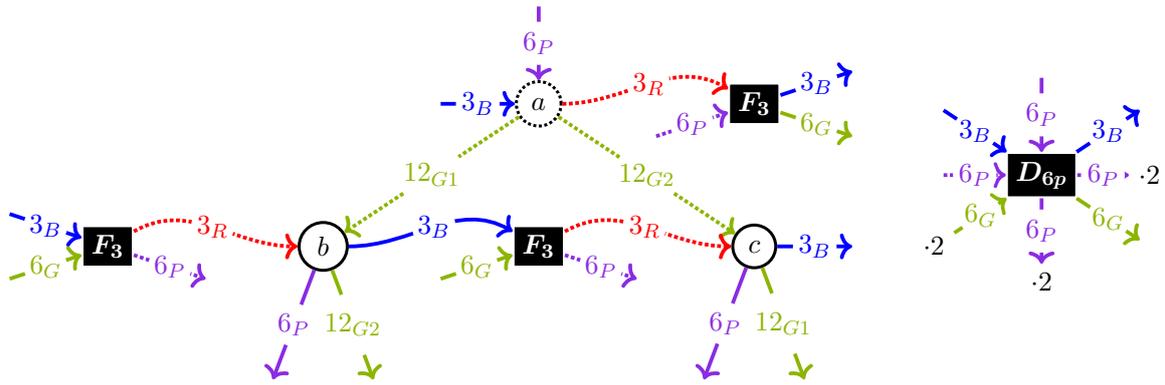
\begin{figure}[hbt]
    \centering
    \resizebox{\textwidth}{!}{    
    \begin{tikzpicture}[scale=2]

        \node[invisible node] (top) at (0, 0.75) {};
        
        \node[invisible node] (left1)           at (-0.75,      0) {};
        \node[graph node B] (D1)                at (0,          0) {$a$};
        \node[invisible node] (F3a in)          at (0.75,       -0.25) {};
        \node[shorthand] (F3a)                  at (1.5,        0) {$F_3$};
        \node[invisible node] (F3a out top)     at (2.25,       0.25) {};
        \node[invisible node] (F3a out bottom)  at (2.25,       -0.25) {};

        \node[invisible node] (F3b in top)      at (-3.75,      -0.75) {};
        \node[invisible node] (F3b in bottom)   at (-3.75,      -1.25) {};
        \node[shorthand] (F3b)                  at (-3,         -1) {$F_3$};
        \node[invisible node] (F3b out)         at (-2.25,      -1.25) {};

        \node[graph node] (D2)                  at (-1.5,       -1) {$b$};
        \node[shorthand] (F3c)                  at (0,          -1) {$F_3$};
        \node[graph node] (D3)                  at (1.5,        -1) {$c$};
        \node[invisible node] (right2)          at (2.25,       -1) {};

        \node[invisible node] (Fc in)    at (-0.75,   -1.25) {};
        \node[invisible node] (Fc out)   at (0.75,   -1.25) {};
        
        \node[invisible node] (low left 1) at (-1.75-0.125,  -2) {};
        \node[invisible node] (low left 2) at (-1.25+0.125,  -2) {};
        \node[invisible node] (low mid 1) at (1.25-0.125,    -2) {};
        \node[invisible node] (low mid 2) at (1.75+0.125,    -2) {};
        
        \path[->] (top) edge[std, ourPurple] node[edge descriptor] {\color{blue-violet}$6_P$} (D1);
        \path[->] (left1) edge[std, ourBlue] node[edge descriptor] {\color{blue}$3_B$} (D1);
        \path[->] (D1) edge[out = 0, in = -210, std B, ourRed] node[edge descriptor] {\color{red}$3_R$} (F3a);
        \path[->] (F3a in) edge[std B, ourPurple] node[edge descriptor] {\color{blue-violet}$6_P$} (F3a);
        \path[->] (F3a) edge[std, ourBlue] node[edge descriptor] {\color{blue}$3_B$} (F3a out top);
        \path[->] (F3a) edge[std, ourGreen] node[edge descriptor] {\color{applegreen}$6_G$} (F3a out bottom);        
        \path[->] (D1) edge[std B, ourGreen] node[edge descriptor] {\color{applegreen}$12_{G1}$} (D2);
        \path[->] (D1) edge[std B, ourGreen] node[edge descriptor] {\color{applegreen}$12_{G2}$} (D3);

        \path[->] (F3b in top) edge[std, ourBlue] node[edge descriptor] {\color{blue}$3_B$} (F3b);
        \path[->] (F3b in bottom) edge[std, ourGreen] node[edge descriptor] {\color{applegreen}$6_G$} (F3b);
        \path[->] (F3b) edge[out = 30, in = 180, std B, ourRed] node[edge descriptor] {\color{red}$3_R$} (D2);
        \path[->] (F3b) edge[std B, ourPurple] node[edge descriptor] {\color{blue-violet}$6_P$} (F3b out);
        \path[->] (D2) edge[out = 0, in = -210, std, ourBlue] node[edge descriptor] {\color{blue}$3_B$} (F3c);
        \path[->] (F3c) edge[out = 30, in = 180, std B, ourRed] node[edge descriptor] {\color{red}$3_R$} (D3);
        \path[->] (D3) edge[std, ourBlue] node[edge descriptor] {\color{blue}$3_B$} (right2);

        \path[->] (Fc in) edge[std, ourGreen] node[edge descriptor] {\color{applegreen}$6_G$} (F3c);
        \path[->] (F3c) edge[std B, ourPurple] node[edge descriptor] {\color{blue-violet}$6_P$} (Fc out);
        
        \path[->] (D2) edge[std, ourPurple] node[edge descriptor] {\color{blue-violet}$6_P$} (low left 1);
        \path[->] (D2) edge[std, ourGreen] node[edge descriptor] {\color{applegreen}$12_{G2}$} (low left 2);
        \path[->] (D3) edge[std, ourPurple] node[edge descriptor] {\color{blue-violet}$6_P$} (low mid 1);
        \path[->] (D3) edge[std, ourGreen] node[edge descriptor] {\color{applegreen}$12_{G1}$} (low mid 2);

        \node[invisible node](short top)        at (3.5, 0.75-0.5){};
        \node[shorthand] (short D)              at (3.5, 0-0.5) {$D_{6p}$};
        \node[invisible node](input top)        at (2.75, 0.5-0.5){};
        \node[invisible node](input mid)        at (2.75, -0.5){};
        \node[invisible node](input bottom)     at (2.75, -0.5-0.5){$\cdot2$};
        \node[invisible node](cons out top)     at (4.25, 0.5-0.5){};
        \node[invisible node](cons out mid)     at (4.25, 0-0.5){$\cdot 2$};
        \node[invisible node](cons out bottom)  at (4.25, -0.5-0.5){};
        \node[invisible node](output)           at (3.5, -0.75-0.5){$\cdot 2$};

        \path[->] (short top) edge[std, ourPurple] node[edge descriptor] {\color{blue-violet}$6_P$} (short D);
        \path[->] (input top) edge[std, ourBlue] node[edge descriptor] {\color{blue}$3_B$} (short D);
        \path[->] (input mid) edge[std B, ourPurple] node[edge descriptor] {\color{blue-violet}$6_P$} (short D);
        \path[->] (input bottom) edge[std, ourGreen] node[edge descriptor] {\color{applegreen}$6_G$} (short D);
        \path[->] (short D) edge[std, ourBlue] node[edge descriptor] {\color{blue}$3_B$} (cons out top);
        \path[->] (short D) edge[std B, ourPurple] node[edge descriptor] {\color{blue-violet}$6_P$} (cons out mid);
        \path[->] (short D) edge[std, ourGreen] node[edge descriptor] {\color{applegreen}$6_G$} (cons out bottom);
        \path[->] (short D) edge[std, ourPurple] node[edge descriptor] {\color{blue-violet}$6_P$} (output);
    
    \end{tikzpicture}
    }
\caption{
A gadget for duplicating $\color{blue-violet}6_P$ input edges, with a shorthand version on the right. The gadget can be extended by adding further layers below to create as many $\color{blue-violet}6_P$ edges as are needed. Constant edges produced by nodes in high layers should be consumed by gadgets in lower layers, but these connections are not shown.
Odd layers will produce female $\color{blue-violet}6_P$ edges and should have $F_3$ flippers to the right of their nodes as shown in upper layer; even layers will produce male $\color{blue-violet}6_P$ edges, and should have $F_3$ flippers to the left of their nodes as shown in the lower layer. 
The figure shows one female node ($a$) and two female nodes ($b, c$) -- this creates an imbalance between the male $\color{applegreen}6_G$ and female $\color{blue-violet}6_P$ nodes consumed and produced and should be avoided by adding nodes until there are as many male nodes as female nodes.
$\color{applegreen}6_G$ edges can be duplicated with a very similar $D_{6g}$ gadget simply by replacing the topmost $\color{blue-violet}6_P$ edge with a $\color{applegreen}6_G$ edge.
}
\label{fig:6-splitter}
\end{figure} 

\begin{lemma}[6-Duplicator gadget schedules]
    \label{lem:D6}
    In any 6-duplicator gadget $D_{6P}$ with an input $\color{blue-violet}6_P$ edge, The non-flipper nodes
    have valid local schedules, which must be of the form 
    \[[\color{red}3_R\color{black}, 
        \color{blue}3_B\color{black}, 
        \color{applegreen}12_{G1}\color{black}, 
        \color{red}3_R\color{black}, 
        \color{blue}3_B\color{black}, 
        \color{blue-violet}6_P\color{black},  
        \color{red}3_R\color{black}, 
        \color{blue}3_B\color{black}, 
        \color{applegreen}12_{G2}\color{black}, 
        \color{red}3_R\color{black},
        \color{blue}3_B\color{black}, \color{blue-violet}6_P\color{black}]
    \]
    in any \slotgood global schedule.
\end{lemma}

\begin{proof}
    Each non-flipper node in $D_{6P}$ has tasks 
    [$\color{red}3_R$, $\color{blue}3_B$, $\color{blue-violet}6_P$, 
    $\color{applegreen}12_{G1}$, $\color{applegreen}12_{G2}$]. It also has
    density $D=1$, so each task with frequency $f$ must appear exactly once every $f$ days. 
    
    Consider node $a$, which has inputs $\color{blue}3_B$ and $\color{blue-violet}6_P$. In any \slotgood schedule these are scheduled in slots of their respective colours, which
    forces partial schedules for $a$ to be of the form 
    $[{\color{red}\SLOT}, 
    {\color{blue}3_B},
    {\color{applegreen}\SLOT},
    {\color{red}\SLOT}, 
    {\color{blue}3_B}, 
    {\color{blue-violet}6_P}]$.
        
    Exactly two of these slots must be filled by $\color{red}3_R$, and if these are not both red slots, the constraint of $\color{red}3_R$ will be violated. Thus, the schedule for $a$ must be of the form
    \[
        [{\color{red}3_R},
         {\color{blue}3_B,} 
         {\color{applegreen}\SLOT},
         {\color{red}3_R}, 
         {\color{blue}3_B},
         {\color{blue-violet}6_P},
         {\color{red}3_R}, 
         {\color{blue}3_B},
         {\color{applegreen}\SLOT},
         {\color{red}3_R},
         {\color{blue}3_B}, 
         {\color{blue-violet}6_P}].
    \]
    Two spaces remain, which must then contain $\color{applegreen}12_{G1}$ and $\color{applegreen}12_{G2}$, as shown in the Lemma.

    Now consider an arbitrary node which is neither $a$ nor a flipper; all such nodes will have inputs
    $\color{applegreen}12_{Ga}$ and $\color{red}3_R$, forcing their partial schedules to be of the form
    \[
       [
        {\color{red}3_R},
        {\color{blue}\SLOT},
        {\color{applegreen}12_{Ga}},
        {\color{red}3_R},
        {\color{blue}\SLOT},
        {\color{blue-violet}\SLOT},
        {\color{red}3_R},
        {\color{blue}\SLOT},
        {\color{applegreen}\SLOT},
        {\color{red}3_R},
        {\color{blue}\SLOT},
        {\color{blue-violet}\SLOT}]
    .\]
    As with $a$, exactly four of these slots must schedule $\color{blue}3_B$ edges, and these must be the blue spaces for the $\color{blue}3_B$ constraint not to be violated, forcing schedules of the form
    \[
        [
        {\color{red}3_R},
        {\color{blue}3_B},
        {\color{applegreen}12_{Ga}},
        {\color{red}3_R},
        {\color{blue}3_B},
        {\color{blue-violet}\SLOT},
        {\color{red}3_R},
        {\color{blue}3_B},
        {\color{applegreen}\SLOT},
        {\color{red}3_R},
        {\color{blue}3_B},
        {\color{blue-violet}\SLOT]}.
    \]
    One of these slots must contain the $\color{applegreen}12_{Gb}$ edge, with the other two scheduling the $\color{blue-violet}6_P$ edge. If the $\color{applegreen}12_{Gb}$ edge is not scheduled in the remaining green slot, the constraint on the $\color{blue-violet}6_P$ edge will be violated, so the schedule must be as shown in the Lemma.
\end{proof}
\begin{corollary}
    Note that by replacing the $\color{blue-violet}6_P$ input edge with a $\color{applegreen}6_G$ input edge, a $D_{6G}$ replicator will be made instead that produces $\color{applegreen}6_G$ edges according to the same logic. 
\end{corollary}
\subsection{Clauses}
\label{sec:OR}

Clauses in $C$ are disjunctions of at most 3 literals, \ie{} the logical OR of at most 3 variables $x_i$, or their negations $\overline{x_i}$. 
A gadget which determines the truth value of a clause given the values of its variables is shown in \wref{fig:OR}. 
\wref{lem:OR} shows that the $12_O$ output edges of this gadget can be blue iff the corresponding clause is evaluated to be True; \wref{def:gadgets-assemble} and \wref{lem:tension} ensure that all $12_O$ edges must be blue in any valid global schedule. Thus, for $\mathcal P_{d\varphi}$ to have a valid schedule, $\varphi$ must have a satisfying assignment.

\begin{figure}[hbt] 
    \centering
    \begin{tikzpicture}[scale=2]
        \node[invisible node] (x1) at (-1, 0) {};
        \node[invisible node] (x2) at (0, 0) {};
        \node[invisible node] (x3) at (1, 0) {};

        \node[graph node B] (I1) at (-1, -1) {$I_1$};
        \node[graph node B] (I2) at (0, -1) {$I_2$};
        \node[graph node B] (I3) at (1, -1) {$I_3$};

        \path[->] (x1) edge[std, ourBlack] node[edge descriptor] {$3_{1R}$} (I1);
        \path[->] (x2) edge[std, ourBlack] node[edge descriptor] {$3_{2R}$} (I2);
        \path[->] (x3) edge[std, ourBlack] node[edge descriptor] {$3_{3R}$} (I3);

        \node[graph node] (OR) at (0, -2) {OR};

        \path[->] (I1) edge[std B, ourBlack] node[edge descriptor] {$12_{1}$} (OR);
        \path[->] (I2) edge[std B, ourBlack] node[edge descriptor] {$12_{2}$} (OR);
        \path[->] (I3) edge[std B, ourBlack] node[edge descriptor] {$12_{3}$} (OR);

        \node[invisible node] (red) at (-1, -2){};

        \path[->] (red) edge[std B, ourRed] node[edge descriptor] {\color{red}$3_R$} (OR);

        \node[graph node B] (fill1) at (1, -1.75){$f_1$};
        \node[graph node B] (fill2) at (1, -2.25){$f_2$};

        \path[->] (OR) edge[std, ourBlack] node[edge descriptor] {$6_{1}$} (fill1);
        \path[->] (OR) edge[std, ourBlack] node[edge descriptor] {$6_{2}$} (fill2);

        \node[invisible node] (out) at (0, -3){};

        \path[->] (OR) edge[std, ourBlack] node[edge descriptor] {$12_O$} (out);

        \node[invisible node] (short x1)  at (2, -0.5) {};
        \node[invisible node] (short x2)  at (3, -0.5) {};
        \node[invisible node] (short x3)  at (4, -0.5) {};
        \node[shorthand]      (short OR)  at (3, -1.5) {OR};
        \node[invisible node] (short r)   at (2, -1.5){};
        \node[invisible node] (short out) at (3, -2.25){};
        
        \path[->] (short x1) edge[std, ourBlack] node[edge descriptor] {$3_{1R}$} (short OR);
        \path[->] (short x2) edge[std, ourBlack] node[edge descriptor] {$3_{2R}$} (short OR);
        \path[->] (short x3) edge[std, ourBlack] node[edge descriptor] {$3_{3R}$} (short OR);
        \path[->] (short r)  edge[std B, ourRed] node[edge descriptor] {\color{red}$3_R$} (short OR);
        \path[->] (short OR) edge[std, ourBlack] node[edge descriptor] {$12_O$} (short out);
    \end{tikzpicture}
    \caption{
    A gadget which computes $x_1 \lor x_2 \lor x_3$ (left), along with a shorthand version (right). To compute $x_1 \lor x_2 \lor \overline{x_3}$, replace the incoming $3_{3R}$ edge with a $3_{3B}$ edge. To instead compute $x_1 \lor x_2$, replace the incoming $3_{3R}$ edge with a male $\color{blue}3_B$ edge. \wref{lem:OR} shows that the $12_O$ edge can be blue iff at least one input is assigned True, while \wref{def:gadgets-assemble} and \wref{lem:tension} show that $12_O$ edges must be blue in any valid global schedule.
    }
    \label{fig:OR}
\end{figure}
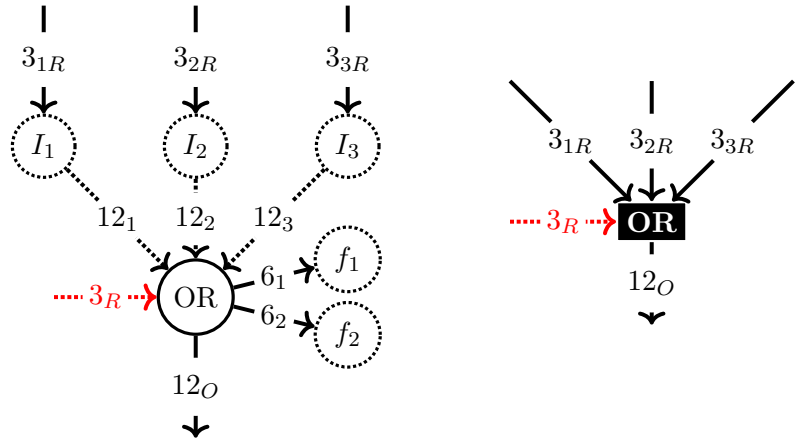 

\begin{lemma}[OR gadget schedules]
    \label{lem:OR}
	Consider an \OR gadget whose input edges $3_{1R}$, $3_{2R}$, and $3_{3R}$ are either red or blue. 
    There exists a \slotgood schedule in which the $12_O$ output edge of this \OR gadget is blue iff at least one of its input edges is red.
\end{lemma} 

While this Lemma only considers the clause $x_1 \vee x_2 \vee x_3$, all other clauses can be handled similarly: For a literal $x_j$, we use $3_{jR}$ as input and for negated literals $\overline {x_j}$, we instead use $3_{jB}$. Clauses with fewer literals also share the same gadget and proof: here, we replace each unused $3_{jR}$ edge with a male $\color{blue}3_B$ edge, effectively replacing clauses such as $x_1 \vee x_2$ with analogous clauses in the form $x_1 \vee x_2 \vee \bot$.

\begin{proof}
    The node labelled \OR has tasks 
    [$\color{red}3_R$,
    $6_1, 6_2, 12_1, 12_2, 12_3, 12_O$] 
    and local density $D=1$, so each task with frequency $f$ must appear exactly once in every $f$ day period. 
    Any \slotgood schedule must assign the $\color{red}3_R$ edge to red slots, so schedules for the \OR node must be of the form
    [\color{red}$3_R$\color{black},
    \color{blue}$\SLOT$\color{black},
    \color{applegreen}$\SLOT$\color{black},
    \color{red}$3_R$\color{black},
    \color{blue}$\SLOT$\color{black},
    \color{blue-violet}$\SLOT$\color{black}].

    Consider an inverter node $I_{i \in \{1,2,3\}}$, which has tasks
    [$3_{iR}$, $12_i$] 
    where $3_{iR}$ is either red or blue by assumption.
    If the input edge of this node, $3_{iR}$, is red, 
    then partial schedules for $I_i$ must be of the form
    [\coloredunderline{red}{$3_{iR}$},
    \color{blue}$\SLOT$\color{black},
    \color{applegreen}$\SLOT$\color{black},
    \coloredunderline{red}{$3_{iR}$},
    \color{blue}$\SLOT$\color{black},
    \color{blue-violet}$\SLOT$\color{black}] 
    and the $12_i$ edge must be either green, purple, or blue.
    Similarly, if $3_{iR}$ is scheduled in blue slots,
    then partial schedules for $I_i$ must instead be of the form
    [\color{red}$\SLOT$\color{black},
    \coloredunderline{blue}{$3_{iR}$},
    \color{applegreen}$\SLOT$\color{black},
    \color{red}$\SLOT$\color{black},
    \coloredunderline{blue}{$3_{iR}$},
    \color{blue-violet}$\SLOT$\color{black}] 
    and the $12_i$ edge must be either green, purple, or red;
    however, this $12_i$ edge is also connected to the \OR node, which has no empty red slots, further restricting it to be either green or purple.

    Suppose that one input edge, $3_{tR}$, is scheduled in red slots, while the others, $3_{iR}$ and $3_{i'R}$, may be scheduled in either red or blue slots. As $3_{tR}$ is red, the output of the associated inverter, $12_{t}$, may be blue, green, or purple.
    Consider the schedule
    \[[\color{red}3_R\color{black},
    \mathunderline{blue}{12_t},
    \mathunderline{applegreen}{6_a},
    \color{red}3_R\color{black},
    \mathunderline{blue}{6_b},
    \mathunderline{blue-violet}{12_i},
    \color{red}3_R\color{black},
    \mathunderline{blue}{12_O},
    \mathunderline{applegreen}{6_a},
    \color{red}3_R\color{black},    
    \mathunderline{blue}{6_b},
    \mathunderline{blue-violet}{12_{i'}}],\]
    observing that $12_O$ is scheduled in blue slots and that no constraints are violated -- a valid schedule.

    Now suppose that all three inputs, $3_{1R}$, $3_{2R}$, and $3_{3R}$, are scheduled in blue slots. 
    By the reasoning above, $12_1$, $12_2$, and $12_3$ must now be scheduled in green or purple slots.
    This will force schedules for the \OR node to be of the form
    \begin{align*}
        &[\color{red}3_R\color{black},
        \color{blue}\SLOT\color{black},
        \mathunderline{applegreen}{12_a},
        \color{red}3_R\color{black},
        \color{blue}\SLOT\color{black},
        \mathunderline{blue-violet}{12_b},
        \color{red}3_R\color{black},
        \color{blue}\SLOT\color{black},
        \mathunderline{applegreen}{12_c},
        \color{red}3_R\color{black},
        \color{blue}\SLOT\color{black},
        \color{blue-violet}\SLOT\color{black}]
        \ \text{or}
        \\
        &[\color{red}3_R\color{black},
        \color{blue}\SLOT\color{black},
        \mathunderline{applegreen}{12_a},
        \color{red}3_R\color{black},
        \color{blue}\SLOT\color{black},
        \mathunderline{blue-violet}{12_b},
        \color{red}3_R\color{black},
        \color{blue}\SLOT\color{black},
        \color{applegreen}\SLOT\color{black},
        \color{red}3_R\color{black},
        \color{blue}\SLOT\color{black},
        \mathunderline{blue-violet}{12_c}],
   \end{align*} 
    for some mapping of $12_1$, $12_2$, and $12_3$ onto $12_a$, $12_b$, and $12_c$.
    Assume towards a contradiction that the $12_O$ edge is blue, observing an immediate constraint violation when scheduling $6_1$ and $6_2$. This demonstrates that the $12_O$ edge cannot be blue and the schedule for the \OR node must be of the form
    \[[\color{red}3_R\color{black},
    \mathunderline{blue}{6_a},
    \mathunderline{applegreen}{12_a},
    \color{red}3_R\color{black},
    \mathunderline{blue}{6_b},
    \mathunderline{blue-violet}{12_b},
    \color{red}3_R\color{black},
    \mathunderline{blue}{6_a},
    \mathunderline{applegreen}{12_c},
    \color{red}3_R\color{black},    
    \mathunderline{blue}{6_b},
    \mathunderline{blue-violet}{12_d}]\]
    (likewise, for some mapping of $12_1$, $12_2$, $12_3$, and $12_O$ onto $12_a$, $12_b$, $12_c$, and $12_d$).
    Thus, if all input edges are blue, then the $12_O$ edge must be either green or purple.
\end{proof} 

\subsection{Tension}
\label{sec:Tension}
$P_{d\varphi}$ is intended to have a valid schedule iff the underlying $\varphi$ is satisfied by some assignment of values to its variables. 
\wref{lem:OR} shows that the $12_O$ output edge of an \OR gadget can be blue iff the associated clause is satisfied by this assignment.
This section introduces Tension gadgets (shown in \wref{fig:Tension}), which ensure that in any global schedule for $P_{d\varphi}$ all $12_O$ edges must be blue, and hence that all clauses must be True.

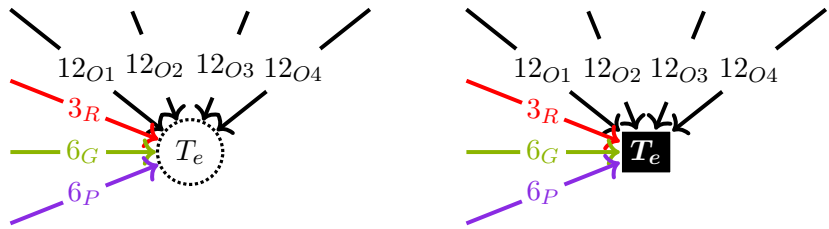
\begin{figure}[hbt] 
    \centering
    \begin{tikzpicture}[scale=2]
        \node[graph node B]   (tension)     at (1.5, -0.75){$T_e$};
        \node[invisible node] (red)         at (0.5-0.25, -0.25) {};
        \node[invisible node] (green)       at (0.5-0.25, -0.75) {};
        \node[invisible node] (purple)      at (0.5-0.25, -1.25) {};
        \node[invisible node] (in 1)        at (0.5-0.25, 0.25) {};
        \node[invisible node] (in 2)        at (1.16667-0.075, 0.25) {};
        \node[invisible node] (in 3)        at (1.83333+0.075, 0.25) {};
        \node[invisible node] (in 4)        at (2.5+0.25, 0.25) {};

        \path[<-] (tension) edge[std, ourBlack] node[edge descriptor] {$12_{O1}$} (in 1);
        \path[<-] (tension) edge[std, ourBlack] node[edge descriptor] {$12_{O2}$} (in 2);
        \path[<-] (tension) edge[std, ourBlack] node[edge descriptor] {$12_{O3}$} (in 3);
        \path[<-] (tension) edge[std, ourBlack] node[edge descriptor] {$12_{O4}$} (in 4);
        \path[->] (red) edge[std, ourRed] node[edge descriptor] {$\color{red}3_R$} (tension);
        \path[->] (green) edge[std, ourGreen] node[edge descriptor] {$\color{applegreen}6_G$} (tension);
        \path[->] (purple) edge[std, ourPurple] node[edge descriptor] {$\color{blue-violet}6_P$} (tension);
        
        \node[shorthand]      (tension)     at (4.5, -0.75){$T_e$};
        \node[invisible node] (red)         at (3.5-0.25, -0.25) {};
        \node[invisible node] (green)       at (3.5-0.25, -0.75) {};
        \node[invisible node] (purple)      at (3.5-0.25, -1.25) {};
        \node[invisible node] (in 1)        at (3.5-0.25, 0.25) {};
        \node[invisible node] (in 2)        at (4.16667-0.075, 0.25) {};
        \node[invisible node] (in 3)        at (4.83333+0.075, 0.25) {};
        \node[invisible node] (in 4)        at (5.5+0.25, 0.25) {};

        \path[<-] (tension) edge[std, ourBlack] node[edge descriptor] {$12_{O1}$} (in 1);
        \path[<-] (tension) edge[std, ourBlack] node[edge descriptor] {$12_{O2}$} (in 2);
        \path[<-] (tension) edge[std, ourBlack] node[edge descriptor] {$12_{O3}$} (in 3);
        \path[<-] (tension) edge[std, ourBlack] node[edge descriptor] {$12_{O4}$} (in 4);
        \path[->] (red) edge[std, ourRed] node[edge descriptor] {$\color{red}3_R$} (tension);
        \path[->] (green) edge[std, ourGreen] node[edge descriptor] {$\color{applegreen}6_G$} (tension);
        \path[->] (purple) edge[std, ourPurple] node[edge descriptor] {$\color{blue-violet}6_P$} (tension);
    \end{tikzpicture}
    \caption{
    A gadget (left) which applies tension to four inputs, ensuring that either $12_{O1}$, $12_{O2}$, $12_{O3}$, and $12_{O4}$ are scheduled in blue slots or no schedule can be found for $T_e$. To apply tension to less than 4 inputs, connect any unneeded inputs to its own male pendant node.
    }
    \label{fig:Tension}
\end{figure} 

\begin{lemma}[Tension gadget schedules]
	\label{lem:tension}
	Each Tension gadget $T_e$ has a valid local schedule iff all $12_O$ input edges are scheduled in blue slots.
\end{lemma}
\begin{proof}
    The single node of each tension gadget $T_e$, as shown in \wref{fig:Tension}, has tasks
    [$\color{red}3_R$,
    $\color{applegreen}6_G$,
    $\color{blue-violet}6_P$,
    $12_{O1}, 12_{O2}, 12_{O3}, 12_{O4}$] 
    and density $D=1$, so each task with frequency $f$ must appear exactly once in every $f$ day period. 
    Further, the $\color{red}3_R$, $\color{applegreen}6_G$, and $\color{blue-violet}6_P$ edges are indirectly connected to the True Clock such that they must be scheduled in slots of their respective colours. Partial schedules for $T_e$ nodes must therefore be of the form
    [$\color{red}3_R$,
    $\color{blue}\SLOT$,
    $\color{applegreen}6_G$,
    $\color{red}3_R$,
    $\color{blue}\SLOT$,
    $\color{blue-violet}6_P$
    ].
    All empty slots in schedules of this form are blue, so the remaining $12_{O1}, 12_{O2}, 12_{O3}, 12_{O4}$ edges must be scheduled in blue slots, creating valid schedules of the form
    \[[\color{red}3_R\color{black},
    \mathunderline{blue}{12_{Oa}}\color{black},
    \color{applegreen}6_G\color{black},
    \color{red}3_R\color{black},
    \mathunderline{blue}{12_{Ob}}\color{black},
    \color{blue-violet}6_P\color{black},
    \color{red}3_R\color{black},
    \mathunderline{blue}{12_{Oc}}\color{black},
    \color{applegreen}6_G\color{black},
    \color{red}3_R\color{black},
    \mathunderline{blue}{12_{Od}}\color{black},
    \color{blue-violet}6_P\color{black}
    ],\]
    for any mapping of $12_{O1}$, $12_{O2}$, $12_{O3}$, and $12_{O4}$ onto $12_{Oa}$, $12_{Ob}$, $12_{Oc}$, and $12_{Od}$.
\end{proof}

\subsection{Proof of Reduction}

With these preparations we can, at long last, prove \wref{lem:dps=3sat}, and hence complete the proof of \wref{thm:sat-inapprox}.

\begin{proofof}{\wref{lem:dps=3sat} on page \pageref{lem:dps=3sat}}
	Consider a polycule $\mathcal P_{d\varphi}$ built from some 3-CNF formula $\varphi = c_1\land \cdots\land c_m$ over variables $X = \{x_1,\ldots, x_{n}\}$ using the algorithm defined by \wref{def:gadgets-assemble}. 
	
    First, suppose that there is a variable assignment $v:X\to\{\mathrm{True},\mathrm{False}\}$ which satisfies all clauses $c_i\in C$.
	We construct a \slotgood schedule $S$ for $\mathcal P_{d\varphi}$ as follows:
\begin{thmenumerate}{lem:dps=3sat}
    \item At the variable gadget for each $x_i\in X$, schedule $3_{iR}$ in red slots if $v(x_i) = \mathrm{True}$ and in the blue slots otherwise.
	This fixes a schedule for all edges in the variable layer.
    \item \wref{lem:D3} and \wref{lem:D6} give schedules for duplication gadgets, whose number should be sufficient to cover the needs of all other gadgets.
    \item As $v$ satisfies all clauses $c_i\in C$, each clause has at least one literal ($x_i$ or $\overline{x_i}$) which evaluates to True. Hence, each clause has one input $3_{iR}$ or $3_{iB}$ edge which will be scheduled in red slots.
    By \wref{lem:OR}, there therefore exists a valid schedule for each \OR gadget in the clause layer such that the $12_O$ output edge of that gadget is scheduled in blue slots.
    \item \wref{lem:tension} shows that this implies a valid schedule for each Tension gadget.
\end{thmenumerate}
    These valid schedules for each layer of gadgets combine to form a global schedule $S$ for $\mathcal P_{d\varphi}$.
	
	Now, assume that we are given a schedule $S$ for $\mathcal P_{d\varphi}$.
	By applying \wref{lem:variable-6-color} to the True Clock, we can assign coloured slots to $S$. It then follows from the construction of $\mathcal P_{d\varphi}$ that $S$ must be \slotgood.
	Set $v(x_i) = \mathrm{True}$ if $S$ schedules $3_{iR}$ in red slots and $v(x_i) = \mathrm{False}$ otherwise.
    \wref{lem:OR} shows that for each $c_{i}\in C$, at least one input edge must be red for the output edge to be blue; \wref{lem:tension} shows that all such output edges must be blue given the existence of $S$. Thus, all clauses $c_{i}\in C$ evaluate to True under $v$, and the claim follows.

	It remains to argue that our reduction can be realised in polynomial time.
	The size of the polycule $\mathcal P_{d\varphi}$ is clearly polynomial in the size of the formula, leaving the gadgets themselves. Variable, \OR, and Tension gadgets have both constant size and constant number. D6 duplicator gadgets are linear in size with respect to the number of $\color{blue-violet}6_P$ and $\color{applegreen}6_G$ edges consumed by the system as a whole, bounded by the number of Tension gadgets, hence by the number of clauses. Blue D3 duplicator gadgets are linear with respect to the number of $\color{blue}3_B$ edges consumed by \OR gadgets, \ie{} the number of clauses. $\color{red}3_R$ edges are consumed by both Tension gadgets and \OR gadgets, but this still leaves the size of their D3 duplicator gadgets bounded by the number of clauses and hence by the size of the formula.
    Thus, it is easy to implement our reduction in polynomial time.
\end{proofof}

\plaincenter{{\bfseries%
	\textcolor{red}{*}%
	\qquad%
	\textcolor{blue}{*}%
	\qquad%
	\textcolor{applegreen}{*}%
	\qquad%
	\textcolor{red}{*}%
	\qquad%
	\textcolor{blue}{*}%
	\qquad%
	\textcolor{blue-violet}{*}%
}}

\section{Conclusion and Open Problems}
In this paper, we revisited the Polyamorous Scheduling problem. Via the observation that for an approximation, the hardest local Bamboo Garden Trimming instances around a single person dominate the solution, we found simple generalizations of existing algorithms for Bamboo Garden Trimming that yield constant-factor approximations for Poly Scheduling.

This paper opens up several avenues for future work. 
The most obvious open problem concerns algorithms with better approximation ratios and tighter lower bounds.
We now know that a $5.24$-approximation is possible and that $(\frac43-\varepsilon)$-approximations are NP-hard, leaving a sizable gap.
Apart from progress on this general question, a natural research direction is to look at restricted classes of polycules. For example, can we obtain better approximation guarantees on graphs with bounded maximum degree~$\Delta$ (beyond the simple results from~\cite{GasieniecSmithWild2024})?
Can we exploit bipartite polycules towards a better approximation?

For the poly density of Decision Poly Scheduling, we showed that $d\le \frac14$ implies feasibility;
as for Pinwheel Scheduling, there are infeasible instances with density $d = \frac56 + \varepsilon$ for any $\varepsilon >0$.
Can we say anything for densities between $\frac14$ and $\frac56$?
In particular, we pose the following question.

\begin{question}[5/6-Poly-Density Question]
	Does every Decision Poly Scheduling instance with poly density $d\le \frac56$ admit a schedule?
\end{question}

Given the tight connection between the problems, it is natural to ask whether Bamboo Garden Trimming can be arbitrarily well approximated, or whether there, too, a hardness-of-approximation result can be shown.

There are also several generalizations of Poly scheduling introduced by \cite{GasieniecSmithWild2024} to which our results could be extended: Fungible Polyamorous Scheduling allows for each person $P$ to have to $s$ meetings per day, while Secure Polyamorous Scheduling allows for meetings to contain cliques of any size instead of only pairs. Can poly density be generalized to these problems? Can it identify classes of instances which do or do not permit solutions? And finally, how well do Reduce-Fastest and other known heuristic perform on them?

	\myacknowledgements
\bibliography{refs}

\clearpage
\appendix

\input{}

\end{document}